\documentclass[12pt]{article}

\usepackage{latexsym,amsmath,amscd,amssymb,graphics}
\usepackage{enumerate,framed,comment}
\usepackage{color}
\usepackage{verbatim}
\usepackage[margin=2cm]{geometry}
\usepackage{amsthm}
\usepackage{epsfig,epstopdf,color}
\usepackage{graphicx, subfigure, psfrag} 

\usepackage[colorlinks]{hyperref}
\usepackage{url}  %
\usepackage[percent]{overpic}

\usepackage[T1]{fontenc}
\usepackage[colorlinks]{hyperref}

\usepackage[all]{xy}

\newcommand{\bfi}{\bfseries\itshape}

\newcommand{\rem}[1]{}
\makeatletter

\@addtoreset{figure}{section}
\def\thefigure{\thesection.\@arabic\c@figure}
\def\fps@figure{h, t}
\@addtoreset{table}{bsection}
\def\thetable{\thesection.\@arabic\c@table}
\def\fps@table{h, t}
\@addtoreset{equation}{section}

\makeatother

\rem{
\textwidth 6.5 truein
\oddsidemargin 0 truein
\evensidemargin .2 truein
\topmargin -.6 truein
\textheight 9.1 in
}

\newcommand{\todo}[1]{\vspace{5 mm}\par \noindent
\framebox{\begin{minipage}[c]{0.95 \textwidth}
\tt #1 \end{minipage}}\vspace{5 mm}\par}

\begin{document}

\newtheorem{theorem}{Theorem}[section]
\newtheorem{definition}[theorem]{Definition}
\newtheorem{lemma}[theorem]{Lemma}
\newtheorem{remark}[theorem]{Remark}
\newtheorem{proposition}[theorem]{Proposition}
\newtheorem{corollary}[theorem]{Corollary}
\newtheorem{example}[theorem]{Example}

\def\below#1#2{\mathrel{\mathop{#1}\limits_{#2}}}



\title{Invariant higher-order variational problems}
\author{
Fran\c{c}ois Gay-Balmaz$^{1}$, Darryl D. Holm$^{2}$, David M. Meier$^{2}$, 
\\Tudor S. Ratiu$^{3}$, Fran\c{c}ois-Xavier Vialard$^{2}$
}
\addtocounter{footnote}{1}
\footnotetext{
Laboratoire de 
M\'et\'eorologie Dynamique, \'Ecole Normale Sup\'erieure/CNRS, Paris, France. 
\texttt{gaybalma@lmd.ens.fr}
\addtocounter{footnote}{1} }
\footnotetext{Department of Mathematics and Institute for Mathematical Sciences, Imperial College, London SW7 2AZ, UK. 
\texttt{d.holm@ic.ac.uk, d.meier09@ic.ac.uk, f.vialard@ic.ac.uk}
\addtocounter{footnote}{1}}
\footnotetext{Section de
Math\'ematiques and Bernoulli Center, \'Ecole Polytechnique F\'ed\'erale de
Lausanne,
CH--1015 Lausanne, Switzerland.
\texttt{tudor.ratiu@epfl.ch}
\addtocounter{footnote}{1} 
}%
\date{{\it\color{red} Fondly remembering our late friend Jerry Marsden} }
\maketitle

\makeatother
\maketitle


\noindent \textbf{AMS Classification:} 

\noindent \textbf{Keywords:} 

\begin{abstract} 

We investigate higher-order geometric $k$-splines for template matching on Lie groups. This is motivated by the need to apply diffeomorphic template matching to a series of images, e.g., in longitudinal studies of Computational Anatomy. Our approach formulates Euler-Poincar\'e theory in higher-order tangent spaces on Lie groups. In particular, we develop the Euler-Poincar\'e formalism for higher-order variational problems that are invariant under Lie group transformations. The theory is then applied to higher-order template matching and the corresponding curves on the Lie group of transformations are shown to satisfy higher-order Euler-Poincar\'{e} equations. The example of $SO(3)$ for template matching on the sphere is presented explicitly. Various cotangent bundle momentum maps emerge naturally that help organize the formulas. We also present Hamiltonian and Hamilton-Ostrogradsky Lie-Poisson formulations of the higher-order Euler-Poincar\'e theory for applications on the Hamiltonian side. 
\end{abstract}
\newpage
\tableofcontents


\section{Introduction}\label{Intro-sec}

\paragraph{The purpose of this paper.} 
This paper provides a method for taking advantage of continuous symmetries in solving Lie group invariant optimization problems for cost functions that are defined on $k^{th}$-order tangent spaces of Lie groups. The type of application we have in mind is, for example, the interpolation and comparison of a series of images in longitudinal studies in a biomedical setting. 

Previous work on the geometric theory of Lagrangian reduction by symmetry on first-order tangent spaces of Lie groups provides a convenient departure point that is generalized here to allow for invariant variational problems formulated on higher-order tangent spaces of Lie groups. It turns out that this generalization may be accomplished as a series of adaptations of previous advances in Euler-Poincar\'e theory, placed into the context of higher-order tangent spaces. Extension of the basic theory presented here to allow for actions of Lie groups on Riemannian manifolds should have several interesting applications, particularly in image registration, but perhaps elsewhere, too. Actions of Lie groups on Riemannian manifolds will be investigated in a subsequent treatment. Two important references for the present work are \cite{HoMaRa1998} for the basic Euler-Poincar\'e theory and \cite{CeMaRa2001} for the bundle setting of geometric mechanics. 

\subsection{Previous work on geometric splines for trajectory planning and interpolation}

The topics treated here fit into a class of problems in control theory called \emph{trajectory planning and interpolation by variational curves}. These problems arise in numerous applications in which velocities, accelerations, and sometimes higher-order derivatives of the interpolation path need to be optimized simultaneously. Trajectory planning using variational curves in Lie groups acting on Riemannian manifolds has been discussed extensively in the literature.  For example, trajectory planning for rigid body motion involves interpolation on either the orthogonal group $SO(3)$ of rotations in $\mathbb{R}^3$, or the semidirect-product group $SE(3)\simeq SO(3)\,\circledS\, \mathbb{R}^3$ of three-dimensional rotations and translations in Euclidean space. Trajectory planning problems have historically found great utility with applications, for example,  in aeronautics, robotics, biomechanics, and air traffic control.  \medskip

Investigations of the trajectory planning problem motivated the introduction in \cite{GaKa1985} and \cite{NoHePa1989} of a class of variational curves called \emph{Riemannian cubics}. Riemannian cubics and their recent higher order generalizations are reviewed in \cite{Popiel2007} and \cite{MaSLeKr2010}, to which we refer for extensive references and historical discussions. The latter work addresses the interpolation by variational curves that generalizes the classical least squares problem to Riemannian manifolds. This generalization is also based on the formulation of higher-order variational problems, whose solutions are smooth curves minimizing the $L^2$-norm of the covariant derivative of order $k \ge 1$, that fit a given data set of points at given times. These solutions are called $k^{th}$-order geometric splines, or geometric $k$-splines. This approach was initiated in \cite{NoHePa1989} for the construction of smoothing splines with $k=2$ for the Lie group $SO(3)$ and
  then generalized to higher order in \cite{CaSLCr1995}. The following result, noted in the first of these papers and then discussed more generally in the second one, was another source of motivation for the present work.  

\begin{proposition} [\cite{NoHePa1989}] \label{Noakes-thm}$\,$ \\ 
The equation for a $2^{nd}$-order geometric spline for a bi-invariant metric on $SO(3)$ may be written as a dynamical equation for a time-dependent vector $\mathbf{\Omega}(t)\in\mathbb{R}^3$ using the vector cross product
\begin{equation}
\dddot{\mathbf{\Omega}}=\ddot{\mathbf{\Omega}}\times\mathbf{\Omega},
\label{Noakes-eqn}
\end{equation}
for all $t$ in a certain interval $[0,T]$.
\end{proposition}
Solutions of the more general version of equation (\ref{Noakes-eqn}) expressed in \cite{CrSL1995} for $2^{nd}$-order geometric splines on Lie groups in terms of the Lie algebra commutator are called `Lie quadratics' in \cite{No2003,No2004,No2006}.

As we said, understanding the intriguing result in Proposition \ref{Noakes-thm} from the viewpoint of Lie group-invariant higher-order variational principles was one of the motivations for the present work. Its general version is proved again below as equation (\ref{CrSLe-commutator}) in Section \ref{hoEP-sec} by using the Euler-Poincar\'e methods of \cite{HoMaRa1998} for higher-order variational principles that are invariant under the action of a Lie group. The directness and simplicity of the present proof of the general version of Proposition \ref{Noakes-thm} compared with other proofs available in the literature encouraged us to continue investigating the application of Lie group-invariant $k^{th}$-order variational principles for geometric $k$-splines. It turns out that higher-order Euler-Poincar\'e theory is the perfect tool for studying  geometric $k$-splines.
 
The Euler-Poincar\'e theory for first-order invariant variational principles focuses on the study of geodesics on Lie groups, which turns out to be the fundamental basis for both ideal fluid dynamics and modern large-deformation image registration. For reviews and references to earlier work on first-order invariant variational principles, see \cite{HoMaRa1998} for ideal fluids and \cite{You2010} for large-deformation image registration. The present paper begins by extending these earlier results for geodesics governed by first-order variational principles that are invariant under a Lie group, so as to include dependence on higher-order tangent spaces of the group (i.e., higher-order time derivatives of curves on the group). This extension is precisely what is needed in designing geometric $k$-splines for trajectory planning problems on Lie groups. The essential strategy in making this extension is the application of reduction by symmetry to the Lagrangian before taking 
variations, as introduced in \cite{HoMaRa1998} for continuum dynamics. The equivalence of the result of Lagrangian reduction by symmetry with the results in the literature for Riemannian cubics and $k^{th}$-order geometric splines is shown in Section \ref{hoEP-sec}, Proposition \ref{Prop-equiv_dash}. 

This previous work has created the potential for many possible applications. In this paper, we shall concentrate on the application of these ideas in template matching for Computational Anatomy (\textit{CA}). Although we do not perform explicit image matching here, we demonstrate the higher-order approach to template matching in the finite dimensional case by interpolating a sequence of points on the sphere $S^2$, using $SO(3)$ as the Lie group of transformations.

\subsection{Main content of the paper}
The {\it main content} of the paper is outlined as follows:
\begin{description}
\item
Section \ref{geomset-sec}  discusses the geometric setting for the present investigation of extensions of group-invariant variational principles to higher order. In particular, Section \ref{geomset-sec} summarizes the definition of higher order tangent bundles and connection-like structures defined on them, mainly by adapting the treatment in \cite{CeMaRa2001} for the geometric formulation of Lagrangian reduction.
\item
Section \ref{hoEP-sec} explains the quotient map for higher-order Lagrangian reduction by symmetry and uses it to derive the basic ${k^{th}}$-order Euler-Poincar\'e equations. This extends to higher-order the Euler-Poincar\'e equations derived in \cite{HoMaRa1998}.  The ${k^{th}}$-order Euler-Poincar\'e equations are then applied to derive the equations for geometric $k$-splines on a Lie group.
After these preliminary developments, there follows a sequence of adaptations of previous advances in Euler-Poincar\'e theory to higher-order tangent spaces. 
\item
Section \ref{Clebsch-sec} extends the Clebsch-Pontryagin approach of \cite{GBRa2010}
to develop the $k^{th}$-order Euler-Poincar\'e equations for potential applications in optimal control.
This extension highlights the role of coadjoint motion for cotangent-lift momentum maps.

\item
Section \ref{hoTemplateMatching-sec} addresses theoretical and numerical results for our main motivation, longitudinal data interpolation. That is, interpolation through a sequence of data points.
After a brief account of the previous work done in Computational Anatomy (\textit{CA}), we derive the equations that generalize the equations for geodesic template matching \cite{BrFGBRa2010} to the case of \emph{higher-order} cost functionals and sequences of \emph{several} data points. We recover in particular the higher-order Euler-Poincar\'{e} equations. For a particular choice of cost functionals one can therefore think of the higher-order template matching approach as \emph{template matching by geometric $k$-splines}. We discuss the gain in smoothness afforded by the higher-order approach, then we provide a qualitative discussion of two Lagrangians that are of interest for applications in \textit{CA}. Finally, we close the section by demonstrating the higher-order approach to template matching in the finite dimensional case by interpolating a sequence of points on the sphere $S^2$, using $SO(3)$ as the Lie group of transformations. This yields the template-matching 
analog of the NHP equation of \cite{NoHePa1989} in 
(\ref{Noakes-eqn}). The results are shown as curves on the sphere in Figures \ref{HOTemplateMatching}. A sample figure is shown below to explain the type of results we obtain. 

\begin{figure}[h!]
\centerline{
\includegraphics[width=0.5\textwidth]{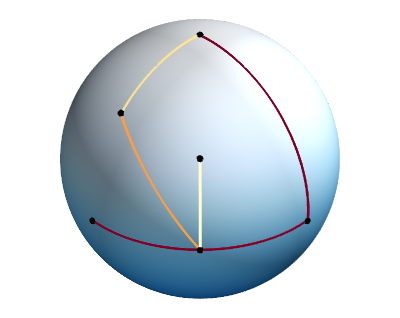}
\includegraphics[width=0.5\textwidth]{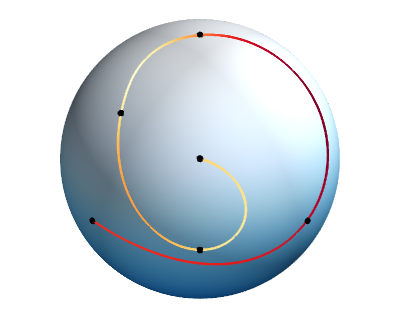}
}
\caption{\footnotesize{
First order \emph{vs.} second order template matching results interpolating a sequence of evenly time-separated points on the sphere, using a bi-invariant metric on the rotation group $SO(3)$. The colors show the local speed along the curves on the spheres (white smaller, red larger). The motion slows as the curve tightens.}}
\label{fig:Intro}
\end{figure}


%
\rem{Section \ref{hoLP-sec} generalizes Lagrange-Poincar\'e reduction \cite{CeMaRa2001} to higher order $G$-invariant Lagrangians by using the theory of higher-order tangent bundles and the connection-like structures defined on them. }
\item
Section \ref{optimizDyn-sec} extends to $k^{th}$-order tangents the metamorphosis approach of \cite{HoTrYo2009} for image registration and the optimization dynamics introduced in \cite{GBHoRa2010}.    
\item
Section \ref{Ham-forms-sec} addresses Hamiltonian and Hamilton-Ostrogradsky formulations of the higher-order Euler-Poincar\'e theory. The Hamilton-Ostrogradsky formulation results in a compound Poisson bracket comprising a sum of canonical and Lie-Poisson brackets. 
\item
Section \ref{outlook-sec} discusses the outlook for future research  and other potential applications of the present approach. These include the formulation of higher-order Lie group invariant variational principles that include both curves on Lie groups and the actions of Lie groups on smooth manifolds, and the formulation of a $k$th-order brachistochrone problem. 
\end{description}

This paper represents only the beginning of our work in this direction. 
The extensions to higher order discussed here demonstrate the unity and versatility of the geometric approach. We hope these methods will be a source of inspiration for future analysis and applications of Lie group reduction of higher-order invariant variational problems.

\section{Geometric setting}\label{geomset-sec}

We shall begin by reviewing the definition of higher order tangent bundles and the connection-like structures defined on them. For more details and explanations of the geometric setting for higher-order variational principles see \cite{CeMaRa2001}.

\subsection{$\mathbf{k^{th}}$-order tangent bundles}

The ${k^{th}}$-order tangent bundle $\tau^{(k)}_Q : T^{(k)}Q
\rightarrow Q$ is defined as the set of equivalence classes of
$C^k$ curves in $Q $ under the equivalence relation that
identifies two given curves $q_i(t), i = 1,
2$, if $q_1(0) = q_2(0) = q_0$ and in any local chart we have 
$q^{(l)}_1(0) =q^{(l)}_2(0)$, for $l = 1, 2,\ldots , k$, where $ q^{(l)}$  denotes the
derivative of order $l$. The equivalence class of the curve $q(t)$ at $q_0 \in Q$ is denoted $[q]_{q_0}^{(k)}$. The projection
\[ 
\tau^{(k)}_Q : T^{(k)}Q \rightarrow Q \quad  \mbox{is given by} \quad
\tau^{(k)}_Q\left([q]_{q_0}^{(k)}\right) = q_0.
\]

It is clear that $T^{(0)}Q = Q$, $T^{(1)}Q = TQ$, and that, for $0 \leq l < k$, there
is a well defined fiber bundle structure
\[ 
\tau^{(l,k)}_Q : T^{(k)}Q \rightarrow T^{(l)}Q, \quad  \mbox{given by} \quad
\tau^{(l,k)}_Q\left([q]_{q_0}^{(k)}\right) = [q]_{q_0}^{(l)}.
\]
Apart from the cases where $k = 0$ and $k=1$, the bundles  $T^{(k)}Q$ are not vector bundles.
The bundle $T^{(2)}Q$ is often denoted $\ddot{Q}$, and is called the {\bfi
second order bundle}.

\begin{remark}{\rm 
We note that $T^{(k)}Q  = J^k_0(\mathbb{R}, Q)$ consists of 
$k$-jets of curves from $\mathbb{R}$ to $Q$ based at 
$ 0 \in \mathbb{R}$, as defined, for example, in  
\cite[\S12.1.2]{Bourbaki1971}.}
\end{remark}

A smooth map $f : M \rightarrow N$ induces a map
\begin{equation}
\label{kth_order_tangent_map} 
T^{(k)}f : T^{(k)}M \rightarrow T^{(k)}N \quad  \mbox{given by} \quad
T^{(k)}f\left([q]_{q_0}^{(k)}\right) := [f\circ q]^{(k)}_{f(q_0)}.
\end{equation} 
In particular, a group action $\Phi : G \times Q \rightarrow Q$ 
naturally lifts to a group action
\begin{equation}
\label{kth_order_action}
\Phi^{(k)} : G \times T^{(k)}Q \rightarrow T^{(k)}Q \quad  \mbox{given by}
\quad \Phi^{(k)}_g\left([q]_{q_0}^{(k)}\right):=T^{(k)}\Phi_g \left( [q]_{q_0}{^{(k)}} \right) = \left[\Phi_g \circ q\right]^ {(k)}_{\Phi_g(q_0) } .
\end{equation} 
This action endows $T^{(k)}Q$ with
a principal $G $-bundle structure. The quotient
$\left( T^{(k)}Q\right) /G$ is a fiber bundle over the base $Q/G$. The class of the element $[q]_{q_0}^{(k)}$ in the
quotient $\left( T^{(k)}Q\right) /G$ is denoted $\left[[q]_{q_0}^{(k)}\right]_G$.

\paragraph{The case of a Lie group.} The $k^{th}$-order tangent bundle $T^{(k)}G$ of  a Lie group $G$ carries a natural Lie group structure: if
$[g]_{g_0}^{(k)}$, and $[h]_{h_0}^{(k)}$ are classes of curves $g$ and
$h$ in $G$, define $[g]_{g_0}^{(k)}[h]_{h_0}^{(k)}:= [gh]_{g_0h_0}^{(k)}$. The Lie algebra $T_eT^{(k)}G$ of $T^{(k)}G$ can be naturally
identified, as a vector space, with $(k + 1)\mathfrak{g}$ (that is, the direct
sum of $k+1$ copies of $\mathfrak{g}$) which, therefore, carries a unique Lie
algebra structure such that this identification becomes a Lie algebra
isomorphism.

\subsection{$\mathbf{k^{th}}$-order Euler-Lagrange equations}\label{sec 2.2}

Consider a Lagrangian $L:T^{ (k) }Q\rightarrow \mathbb{R}  $, $L=L \left(  q, \dot q, \ddot q, ..., q^{ (k) }\right)$. Then a curve
$q: [t_0, t_1] \rightarrow Q$
is a critical curve of the action
\begin{equation}\label{Euler-Lagrange_action}
 \mathcal{J}[q] = \int_{t_0}^{t _1} L \left( q(t), \dot q(t),....,  q^{ (k) }(t) \right)  dt  
\end{equation}
among all curves  $q(t) \in Q$ whose first $(k-1)$ derivatives are fixed at the endpoints: $q^{(j)}(t_i)$, $i=0,1$, $j=0,..., k-1$,
if and only if $q(t)$ is a solution of the $k^{th}$-order Euler-Lagrange equations
  \begin{equation}\label{EL-eqns}
    \sum_{j=0}^k (-1)^j \frac{d^j}{dt^j}  \frac{\partial  L}{\partial  q^{ (j) }}=0.
  \end{equation}
The corresponding variational principle is Hamilton's principle,
\[
\delta \int_{t _0 }^{ t _1} L\left( q(t), \dot q(t),....,  q^{ (k) }(t) \right)  dt  =0
.
\]
In the $\delta$-notation, an infinitesimal variation of the curve $q(t)$ is denoted by $\delta q(t) $ and defined by the \emph{variational derivative}, 
\begin{equation}
\label{var-deriv-def}
\delta q(t) : 
=\left.\frac{d}{d \varepsilon}\right|_{\varepsilon =0 } 
q(t, \varepsilon),
\end{equation}
where $q(t,0)=q(t)$ for all $t$ for which the curve is 
defined and $\frac{\partial^j q}{\partial t^j}(t_i,\varepsilon) = q^{(j)}(t_i)$, for
all $\varepsilon$, $j=0,1,\ldots, k-1$, $i=0,1$. Thus
$\delta q ^{ (j) }(t_0)=0=\delta q ^{ (j) }(t_1)$ for $j=0,...,k-1$. Note that the local notation $L\left( q, \dot q,....,  q^{ (k) }\right)$ used above can be intrinsically written as $L\left(\left[q\right]^{(k)}_q\right)$.

\paragraph{Examples: Riemannian cubic polynomials and generalizations.}
As originally introduced in \cite{NoHePa1989}, Riemannian cubic
polynomials generalize Euclidean splines to Riemannian manifolds.
Let $(Q,\gamma )$ be a Riemannian manifold and $\frac{D}{Dt}$ be the
covariant derivative along curves associated with the Levi-Civita
connection $\nabla$ for the metric $\gamma$.
The Riemannian cubic polynomials are defined as minimizers of the
functional $\mathcal{J}$ in \eqref{Euler-Lagrange_action} for the Lagrangian $L: T^{(2)}Q\rightarrow \mathbb{R}$ defined by
\begin{equation}
L(q,\dot{q},\ddot{q}) := \frac{1}{2}\gamma_q\left(\frac{D}{Dt}\dot
q,\frac{D}{Dt}\dot q\right).
\end{equation}
This Lagrangian is well-defined on the second-order tangent bundle since, in coordinates
\begin{equation}
\frac{D}{Dt}\dot{q}^k=
\ddot{q}^k + \Gamma_{ij}^k(q) \dot{q}^i\dot{q}^j ,
\end{equation}
where $(\Gamma_{ij}^k(q))_{i,j,k}$ are the Christoffel symbols at point
$q$ of the metric $\gamma$ in the given basis.
These Riemannian cubic polynomials have been generalized to the
so-called elastic splines through the following class of Lagrangians
\begin{equation}\label{2_splines_Lagr}
L_{\tau}(q,\dot{q},\ddot{q}) : = \frac{1}{2} \gamma_{q} \left(
\frac{D}{Dt}\dot q, \frac{D}{Dt}\dot q\right)  + \frac{\tau^2}{2}  \gamma_{q}
(\dot{q},\dot{q}) ,
\end{equation}
where $\tau$ is a real constant.
Another extension are the higher-order Riemannian splines, or geometric $k$-splines, where
\begin{equation}\label{k_splines_Lagr}
L_{k}\left( q,\dot q,..., q^{(k)} \right) :=  \frac{1}{2} \gamma_{q}
\left( \frac{D^{k-1}}{Dt^{k-1}}\dot q,\frac{D^{k-1}}{Dt^{k-1}}\dot q
\right) ,
\end{equation}
for $k >2$.
As for the Riemannian cubic splines, $L_k$ is well-defined on $T^{(k)}Q$.
Denoting by $R$ the curvature tensor defined as
$R(X,Y)Z= \nabla _X \nabla _Y X- \nabla _Y \nabla _X Z- \nabla
_{[X,Y]}Z$, the Euler-Lagrange equation for elastic splines ($k=2$) reads
\begin{equation} \label{ELeqns-T2}
\frac{D^3}{Dt^3}\dot q(t)+ R \left( \frac{D}{Dt}\dot q(t),\dot q(t)
\right) \dot q(t)= \tau ^2 \frac{D}{Dt}\dot q (t) ,
\end{equation}
as proven in \cite{NoHePa1989}. For the higher-order Lagrangians $L_k$,
the Euler-Lagrange equations read \cite{CaSLCr1995}
\begin{equation}
\frac{D^{2k-1}}{Dt^{2k-1}}\dot q(t) + \sum_{j=2}^k (-1)^j R
\left(\frac{D^{2k-j-1}}{Dt^{2k-j-1}}\dot q(t),
\frac{D^{j-2}}{Dt^{j-2}}\dot q(t)\right)\dot q(t) =0.
\end{equation}

These various Lagrangians can be used to interpolate between given
configurations on $T^{(k)}Q$. The choice of Lagrangian will depend on the application one has in mind. For
instance, the following interpolation problem was addressed in
\cite{HuBl2004} and was motivated by applications in space-based interferometric
imaging.

\paragraph{Interpolation problem.}
\textit{
Given $N+1$ points $q_i\in Q$, $i=0,...,N$ and tangent vectors $v_j \in T_{q_j}Q$, $j = 0,N$, minimize
\begin{equation} \label{ElasticSplines}
\mathcal{J} [q]:= \frac{1}{2} \int_{t_0}^{t_N} \left( \gamma 
_{q(t)}\left( \frac{D}{Dt} \dot{ q} (t) , \frac{D}{Dt} \dot{ q} (t)
\right) + \tau ^2 \gamma  _{q(t)}\left( \dot{ q}(t) , \dot{ q} (t)
\right) \right) dt  ,
\end{equation}
among curves $ t  \mapsto q(t)\in Q$ that are $C^1$ on $[t _0 , t_N ]$, smooth on
$[t_i, t_{i+1}]$, $t_0\leq t_1\leq \ldots \leq t_N$, and subject to the
interpolation constraints
\[
q(t_i)=q_i, \quad\text{for all }i=1,\ldots, N-1
\]
and the boundary conditions
\[
q(t_0)=q_0,\quad \dot q(t_0)= v_0 ,\quad\text{and}\quad
q(t_N)=q_N,\quad \dot q(t_N)= v_N .
\]}In the context of a group action and invariant Lagrangians, we refer
the reader to Section \ref{hoTemplateMatching-sec} for an example of
higher-order interpolation particularly relevant for
\textit{Computational Anatomy}.


\subsection{Quotient space and reduced Lagrangian}

When one deals with a Lagrangian $L: T^{(k)}Q \rightarrow \mathbb{R}$ that is invariant with respect to the lift $\Phi^{(k)}: G \times T^{(k)}Q \rightarrow T^{(k)}Q$ of a group action $\Phi: G \times Q \rightarrow Q$, then the invariance can be exploited to define a new function called the \emph{reduced Lagrangian} on the quotient space $\left(T^{(k)}Q\right)/G$. We review this procedure here. Since this paper mainly deals with the case where $Q = G$, we begin by describing this special case.

Let $G$ be a Lie group and $h \in G$. The right-, respectively left-actions by $h$ on $G$,
\[
R_h: G \rightarrow G, \quad g \mapsto gh, \quad \mbox{and} \quad L_h: G \rightarrow G, \quad g \mapsto hg,
\]
can be naturally lifted to actions on the $k^{th}$-order tangent bundle $T^{(k)}G$ (see \eqref{kth_order_action}). We will denote these lifted actions by concatenation, as in
\begin{align*}
  R_h^{(k)}&: T^{(k)}G \rightarrow T^{(k)}G, \quad [g]_{g_0}^{(k)}\mapsto  R_h^{(k)}\left( [g]_{g_0}^{(k)}\right) =:  [g]_{g_0}^{(k)}h , \quad \mbox{and}\\
L_h^{(k)}&: T^{(k)}G \rightarrow T^{(k)}G, \quad [g]_{g_0}^{(k)}\mapsto  L_h^{(k)}\left( [g]_{g_0}^{(k)}\right) =: h [g]_{g_0}^{(k)} .
\end{align*}
Consider a Lagrangian $L: T^{(k)}G \rightarrow \mathbb{R}$ that is right-, or left-invariant, i.e., invariant with respect to the lifted right-, or left-actions of $G$ on itself. For any $[g]_{g_0}^{(k)}\in T^{(k)}G$ we then get
\begin{equation}
L\left([g]_{g_0}^{(k)}\right) = L|_{T_e^{(k)}G}\left([g]_{g_0}^{(k)}g_0^{-1}\right), \quad \mbox{or} \quad  L\left([g]_{g_0}^{(k)}\right) = L|_{T_e^{(k)}G}\left(g_0^{-1}[g]_{g_0}^{(k)}\right),
\end{equation}
respectively. The restriction $L|_{T_e^{(k)}G}$ of the Lagrangian to the $k^{th}$-order tangent space at the identity $e$ therefore fully specifies the Lagrangian $L$. Moreover, there are natural identifications $\alpha_k: T_e^{(k)}G \rightarrow k \mathfrak{g}$ given by 
\begin{equation}\label{Alpha_k_right_inv}
  \alpha_k\left([g]_{e}^{(k)}\right) := \left(\dot{g}(0), \left.\frac{d}{dt}\right|_{t = 0} \dot{g}(t)g(t)^{-1}, \ldots ,  \left.\frac{d^{k-1}}{dt^{k-1}}\right|_{t = 0} \dot{g}(t)g(t)^{-1}\right),
\end{equation}
or
\begin{equation}\label{Alpha_k_left_inv}
  \alpha_k\left([g]_{e}^{(k)}\right) := \left(\dot{g}(0), \left.\frac{d}{dt}\right|_{t = 0}g(t)^{-1} \dot{g}(t), \ldots ,  \left.\frac{d^{k-1}}{dt^{k-1}}\right|_{t = 0} g(t)^{-1}\dot{g}(t)\right),
\end{equation}
respectively, where $t \mapsto g(t)$ is an arbitrary representative of $[g]_{e}^{(k)}$.

The reduced Lagrangian $\ell:k\mathfrak{g}\rightarrow\mathbb{R}$ is then defined as
\begin{equation}
\ell:= L|_{T_e^{(k)}G} \circ \alpha_k^{-1},
\end{equation}
where one uses the choice for $\alpha_k$ that is appropriate, namely \eqref{Alpha_k_right_inv} for a right-invariant Lagrangian $L$ and \eqref{Alpha_k_left_inv} for a left-invariant Lagrangian $L$.
Let $t \mapsto g(t) \in G$ be a curve on the Lie group. For every $t$ this curve defines an element in $T_{g(t)}^{(k)}G$, namely
\begin{equation}
  [g]_{g(t)}^{(k)}:= [h]_{g(t)}^{(k)}, \quad \mbox{where $h$ is the curve} \quad \tau \mapsto h(\tau):= g(t+\tau).
\end{equation}
Note that for the case $k=1$ we write, as usual, $\dot{g}(t) := [g]_{g(t)}^{(1)}$. The following lemma is a direct consequence of the definitions:

\begin{lemma}\label{Lemma_Invariant_Lagrangian}
   Let $t \mapsto g(t)$ be a curve in $G$ and $L: T^{(k)}G \rightarrow \mathbb{R}$ a right-, or left-invariant Lagrangian. Then the following equation holds for any time $t_0$,
\begin{equation}\label{Reduced_curve}
L\left([g]_{g(t_0)}^{(k)}\right) = \ell\left(\xi(t_0), \dot{\xi}(t_0), \ldots , \xi^{(k-1)}(t_0)\right),
\end{equation}
where $\xi := \dot{g} g^{-1}$, or $\xi := g^{-1} \dot{g}$ respectively.
\end{lemma}

This last equation will play a key role in the higher-order Euler-Poincar\'{e} reduction discussed in the next section.

\section{Higher-order Euler-Poincar\'e reduction}\label{hoEP-sec}
In this section we derive the basic $k^{th}$-order Euler-Poincar\'e equations by reducing the variational principle associated to the Euler-Lagrange equations on $T^{(k)}Q$. The equations adopt a factorized form, in which the Euler-Poincar\'e operator at $k=1$ is applied to the Euler-Lagrange operation acting on the reduced Lagrangian $\ell ( \xi, \dot \xi, \ddot{\xi},\dots,  \xi^{ (k-1) } ) : k \mathfrak{g}  \rightarrow \mathbb{R} $ at the given order, $k$. We then apply the $k^{th}$-order Euler-Poincar\'e equations to derive the equations for geometric $k$-splines. 

\subsection{Quotient map, variations and $\mathbf{k^{th}}$-order Euler-Poincar\'e equations}\label{k_order_EP_sec}

Let $L: T^{(k)}G \rightarrow \mathbb{R}$ be a right-, or left-invariant Lagrangian. Recall from \S\ref{sec 2.2} that the Euler-Lagrange equations are equivalent to the following variational problem: \newline
\textit{For given $h_i \in G$ and $[h]_i^{(k-1)} \in 
T^{(k-1)}_{h_i}G$, $i = 1, 2$, find a critical curve of the functional 
\[
\mathcal{J}[g] = \int_{t_1}^{t_2} L\left([g]_{g(t)}^{(k)}\right) dt
\]
among all curves $g: t \in [t_1, t_2] \mapsto g(t) \in G$ satisfying the endpoint condition
\begin{equation}\label{kth_order_endpoint_condition}
[g]_{g(t_i)}^{(k-1)} = [h]_i^{(k-1)}, \quad i= 1, 2.
\end{equation}}

\noindent The time derivatives of up to order $k-1$ are 
therefore fixed at the endpoints, i.e., $[g]_{g(t_i)}^{(j)} = [h]_i^{(j)}$, $j=0,\ldots,k-1$, are automatically verified. Let $g: t \mapsto g(t) 
\in G$ be a curve and $(\varepsilon,t) \mapsto 
g_\varepsilon(t) \in G$ a variation of $g$ respecting 
\eqref{kth_order_endpoint_condition}. We recall from 
Lemma \ref{Lemma_Invariant_Lagrangian} that, for any 
$\varepsilon$ and any $t_0$,
\begin{equation}\label{Reduced_curve_recall}
L\left([g_\varepsilon]_{g_\varepsilon(t_0)}^{(k)}\right) = \ell\left(\xi_\varepsilon(t_0), \ldots , \xi^{(k-1)}_\varepsilon(t_0)\right),
\end{equation}
where $\xi_\varepsilon := \dot{g_\varepsilon}g_\varepsilon^{-1}$, or $\xi_\varepsilon := g_\varepsilon^{-1} \dot{g_\varepsilon}$ respectively for the right-, or left-invariant Lagrangian $L$. The variation $\delta \xi$ induced by the variation $\delta g$ is given by 
\begin{equation}
\label{constr_var}
\delta \xi = \dot{\eta} \mp [\xi, \eta]
,
\end{equation}
where $\eta := (\delta g) g^{-1}$, or 
$\eta := g^{-1} (\delta g)$, respectively. It follows from the endpoint conditions \eqref{kth_order_endpoint_condition} that $\eta(t_i) = \dot{\eta}(t_i) = \ldots = \eta^{(k-1)}(t_i)= 0$ and therefore $\delta \xi(t_i) =\ldots = \partial_t^{k-2}\delta \xi(t_i) = 0$, for $i = 1, 2$. We are now ready to compute the variation of $\mathcal{J}$:
\begin{align*}
\delta \int_{t_1}^{t_2} L\left([g]_{g(t)}^{(k)}\right) dt &= \left.\frac{d}{d\varepsilon}\right|_{\varepsilon = 0} \int_{t_1}^{t_2} L\left([g_\varepsilon]_{g_\varepsilon(t)}^{(k)}\right) dt
  \stackrel{\eqref{Reduced_curve_recall}}{=} \left.\frac{d}{d\varepsilon}\right|_{\varepsilon = 0} \int_{t_1}^{t_2} \ell\left(\xi_\varepsilon, \ldots, \xi_\varepsilon^{(k-1)}\right) dt\\
  &= \sum _{ j=0}^ {k-1} \int_{t _1 } ^{ t _2 }\left\langle \frac{\delta \ell }{\delta  \xi^{ (j) }}, \delta  \xi^{ (j) } \right\rangle dt
  = \sum _{ j=0}^ {k-1} \int_{t _1 } ^{ t _2 }\left\langle \frac{\delta \ell }{\delta  \xi^{ (j) } }, \partial _t ^j \delta  \xi  \right\rangle dt\\
&= \int_{t _1 } ^{ t _2 }\left\langle \sum _{ j=0}^ {k-1} (-1)^j \partial _t ^j\frac{\delta \ell }{\delta  \xi^{ (j) } },  \delta  \xi \right\rangle dt\\
&= \int_{t _1 } ^{ t _2 }\left\langle \sum _{ j=0}^ {k-1} (-1)^j \partial _t ^j\frac{\delta \ell }{\delta  \xi^{ (j) } },  \partial _t \eta \mp [ \xi , \eta  ]  \right\rangle dt\\
&= \int_{t _1 } ^{ t _2 }\left\langle \left( - \partial _t \mp \operatorname{ad}^*_ \xi \right)  \sum _{ j=0}^ {k-1} (-1)^j \partial _t ^j\frac{\delta \ell }{\delta  \xi^{ (j) } },  \eta   \right\rangle dt,
\end{align*}
were we used the vanishing endpoint conditions $\delta \xi(t_i) =\ldots = \partial_t^{k-2}\delta \xi(t_i) = 0$ and $\eta(t_i) = 0$, for $i = 1, 2$, when integrating by parts. Therefore, the stationarity condition $\delta \mathcal{J} = 0$ implies the {\bfi $\mathbf{k^{th}}$-order Euler-Poincar\'e equation},
\begin{framed}
\begin{equation}\label{EP_k_dash}
\left(\partial _t  \pm \operatorname{ad}^*_ \xi \right)  
\sum _{ j=0}^ {k-1} (-1)^j \partial _t ^j  \frac{\delta \ell }{\delta  \xi^{ (j) } }=0.
\end{equation}
\end{framed}

\noindent Formula (\ref{EP_k_dash}) takes the following forms for various choices of $k=1,2,3$:\\
\noindent
If $k=1$:
\[
\left( \partial _t  \pm \operatorname{ad}^*_ \xi \right) \frac{\delta \ell }{\delta  \xi}=0
,\]
If $k=2$:
\begin{equation}\label{EP_2}
\left( \partial _t \pm \operatorname{ad}^*_ \xi \right) 
\left( \frac{\delta  \ell}{\delta \xi } - \partial _t  \frac{\delta \ell }{\delta \dot \xi } \right) =0
,\end{equation}
If $k=3$:
\[
\left(\partial _t  \pm \operatorname{ad}^*_ \xi \right) \left( \frac{\delta  \ell}{\delta \xi } 
- \partial _t  \frac{\delta \ell }{\delta \dot \xi } 
+ \partial _t ^2  \frac{\delta \ell }{\delta \ddot \xi }\right) =0
.\]
The first of these is the usual Euler-Poincar\'e equation. The others adopt a \emph{factorized} form in which the Euler-Poincar\'e operator $(\partial _t \pm \operatorname{ad}^*_ \xi )$ is applied to the Euler-Lagrange operation on the reduced Lagrangian $\ell ( \xi, \dot \xi, \ddot{\xi},... ) $ at the given order.

The results obtained above are summarized in the following theorem.

\begin{theorem}[$k^{th}$-order Euler-Poincar\'e reduction] Let $L:T^{(k)}G \rightarrow \mathbb{R}  $ be a $G$-invariant Lagrangian and let $\ell: k \mathfrak{g}  \rightarrow \mathbb{R} $ be the associated reduced Lagrangian. Let $g(t)$ be a curve in $G$ and $ \xi (t)= \dot g(t)g(t)^{-1}$, resp. $ \xi (t)= g(t)^{-1}\dot g(t)$ be the reduced curve in the Lie algebra $ \mathfrak{g} $. Then the following assertions are equivalent.
\begin{itemize}
\item[\rm (i)] The curve $g(t)$ is a solution of the $k^{th}$-order Euler-Lagrange equations for $L:T^{(k)}G \rightarrow \mathbb{R}$.
\item[\rm (ii)] Hamilton's variational principle
\[\delta \int_{t _1 }^{t _2 } L\left( g, \dot g, ..., g^{ (k) } \right) dt =0
\]
holds upon using variations $ \delta g$ such that $\delta g ^{ (j) }$ vanish at the endpoints for $j=0,...,k-1$.
\item[\rm (iii)] The $k^{th}$-order Euler-Poincar\'e equations for $\ell: k \mathfrak{g} \rightarrow \mathbb{R}$:
\begin{equation}
\left(\partial _t  \pm \operatorname{ad}^*_ \xi \right)  
\sum _{ j=0}^ {k-1} (-1)^j \partial _t ^j  \frac{\delta \ell }{\delta  \xi^{ (j) } }=0.
\end{equation}
\item[\rm (iv)] The constrained variational principle
\[
\delta \int_{t _1 }^{t _2 } \ell \left(  \xi , \dot \xi ,..., \xi ^ { (k) } \right) =0
\]
holds for constrained variations of the form $ \delta \xi = \partial _t \eta \mp [ \xi, \eta  ]$, where $ \eta $ is an arbitrary curve in $\mathfrak{g}  $ such that $\eta ^{ (j) }$ vanish at the endpoints, for all $j=0,..., k-1$.
\end{itemize}

\end{theorem}

\medskip

\begin{remark} {\rm The quotient map \eqref{Alpha_k_right_inv}, respectively \eqref{Alpha_k_left_inv}, can be used for any Lie group $G$. In the case of matrix groups, one might consider the alternative quotient map of the form
\begin{equation}
\left(  g, \dot g,..., g^{(k)}\right) \rightarrow \left(  \nu _1,..., \nu _k\right) ,\quad \nu _j:= g^{(j) } g^{-1}\quad\text{respectively}\quad  \nu _j:= g^{-1}g^{(j) }.
\label{quotmap.alt}
\end{equation}
One may easily pass from the variables $\left( \xi, \dot \xi ,..., \partial _t ^{ (k-1) } \xi \right)$ to the variables $\left(  \nu _1,..., \nu _k\right)$. For example:
\begin{align} 
\xi &= \nu _1 \nonumber\\
\dot \xi &= \partial _t ( \dot g g ^{-1} ) = \ddot g g ^{-1} - \dot g g^{-1} \dot g g^{-1} = \nu _2 - \nu _1 \nu _1\\
\ddot \xi &= \nu _3 - 2\nu _2 \nu _1 + 2 \nu _1 \nu _1 \nu _1- \nu _1 \nu _2,
\nonumber
\end{align}
and so forth, by using the rule $ \dot \nu _j =\nu _{j+1}- \nu _j \nu _1 .$
Here all concatenations mean matrix multiplications. One can easily derive the constrained variations and the $k^{th}$-order Euler-Poincar\'e equations associated to this quotient map in a similar way as above. }
\end{remark}

\subsection{Example: Riemannian cubics}\label{2_splines}

In this section we apply the $k^{th}$-order Euler-Poincar\'e reduction to the particular case of $2$-splines on Lie groups. Fix a right-, respectively left-invariant Riemannian metric $\gamma $ on the Lie group $G$. We denote by
\[
\|v_g\|^2_g:=\gamma_g(v_g,v_g)
\]
the corresponding squared norm of a vector $v_g\in T_gG$.
The inner product induced on the Lie algebra $\mathfrak{g}$ is also denoted by $ \gamma: \mathfrak{g}  \times \mathfrak{g}  \rightarrow \mathbb{R}  $ and its squared norm by
\[
\|\xi\|^2_{\mathfrak{g}}:=\gamma(\xi,\xi).
\]
We recall that the associated isomorphisms
 \begin{equation}\label{music-iso-def_dash}
\flat: \mathfrak{g}\to \mathfrak{g}^*, \quad \xi \mapsto \xi^\flat,\quad\text{and}\quad  \sharp: \mathfrak{g}^*\to \mathfrak{g},\quad \mu \mapsto \mu^\sharp,
\end{equation}
are defined by
\begin{equation}
\left\langle\xi^\flat, \eta\right\rangle=\gamma(\xi, \eta), \quad \text{for all}\quad \xi, \eta \in \mathfrak{g},\quad\text{and}\quad \sharp:=\flat^{-1},
\end{equation}
where $\left\langle\,,\right\rangle$ denotes the dual pairing between $\mathfrak{g}^*$ and $\mathfrak{g}$.

\begin{proposition}\label{Prop-equiv_dash} Consider the Lagrangian $L: T^ { (2)} G\rightarrow \mathbb{R}$ for geometric $2$-splines, given by
\begin{equation}
\label{un-reduced-lag-ell_dash}
L(g, \dot{ g}, \ddot{ g})=\frac{1}{2} \left \| \frac{D}{Dt}\dot g\right\|_g^2,
\end{equation}
where $\|\cdot \|$ is the norm of a right-, respectively left-invariant metric on $G$. Then $L$ is right-, respectively left-invariant and induces the reduced Lagrangian $\ell: 2 \mathfrak{g}  \rightarrow \mathbb{R}  $ given by
\begin{equation}
\label{reduced-lag-ell_dash}
\ell( \xi , \dot{ \xi })= \frac{1}{2}\left \| \dot{ \xi } \pm \operatorname{ad}^\dagger_ \xi \xi \right \|_{\mathfrak{g}} ^2,
\end{equation}
where $\operatorname{ad}^\dagger$ by $\operatorname{ad}^\dagger_\xi\eta := \left(\operatorname{ad}^*_\xi(\eta^\flat)\right)^\sharp$, for any $\xi, \eta \in \mathfrak{g}$.
\end{proposition}

\rem{\begin{proposition}\label{Prop-equiv}
Consider the Lagrangian $\ell: T\mathfrak{g}\to\mathbb{R}$ given by the squared norm,
\begin{equation}
\label{reduced-lag-ell}
\ell( \xi , \dot{ \xi })= \frac{1}{2}\left \| \dot{ \xi }^\flat \pm \operatorname{ad}^*_ \xi \xi^\flat \right \| ^2.
\end{equation}
This $\ell$ is the reduced Lagrangian associated to the second order Lagrangian $L: T^ { (2)} G\rightarrow \mathbb{R}$ given by
\begin{equation}
\label{un-reduced-lag-ell}
L(g, \dot{ g}, \ddot{ g})=\frac{1}{2} \left \| \frac{D}{Dt}\dot g\right\|^2,
\end{equation}
where $\|\cdot\|$ is the norm associated the $G$-invariant Riemannian metric induced by $\gamma $, and
\begin{equation}
\label{DDt-g-dot}
\frac{D}{Dt}\dot g= \ddot{ g}+ \Gamma (g) ( \dot{ g}, \dot{ g})
\end{equation}
locally.
\end{proposition}\noindent}

\begin{proof}
Let us recall the expression of the Levi-Civita covariant derivative associated to a right (respectively left) 
$G$-invariant Riemannian metric on $G$. For 
$X\in \mathfrak{X}(G)$ and $v _g \in T_gG$, we have 
(e.g., \cite{KrMi1997}, Section 46.5)
\begin{align}
\label{cov_left}
\nabla_{v_g}X(g)&=TR_g\left(\mathbf{d}f(v_g)+
\frac{1}{2}\operatorname{ad}^\dagger_vf(g)+
\frac{1}{2}\operatorname{ad}^\dagger_{f(g)}v-
\frac{1}{2}[v,f(g)]\right),\quad v:=v_g g ^{-1}\\
\label{cov_right}
\text{resp.}\quad \nabla_{v_g}X(g)&=
TL_g\left(\mathbf{d}f(v_g)-
\frac{1}{2}\operatorname{ad}^\dagger_vf(g)-
\frac{1}{2}\operatorname{ad}^\dagger_{f(g)}v+
\frac{1}{2}[v,f(g)]\right), \quad v:=g ^{-1}v_g
\end{align}
where $f\in \mathcal{F} (G; \mathfrak{g})$
is uniquely determined by the condition $X(g)=TR_g(f(g))$
for right-, respectively $X(g)=TL_g(f(g))$
for left $G$-invariance.
Therefore, we have
\[
\frac{D}{Dt} \dot g(t)= \nabla _{\dot g} \dot {g}
=TR_g\left(  \dot{ \xi } +\frac{1}{2}\operatorname{ad}^\dagger_\xi \xi +\frac{1}{2}\operatorname{ad}^\dagger_{\xi }\xi -\frac{1}{2}[\xi ,\xi ]\right) 
=TR_g\left(  \dot{ \xi } +\operatorname{ad}^\dagger_\xi \xi \right),
\]
respectively
\[
\frac{D}{Dt} \dot g(t)= \nabla _{\dot g} \dot {g}
=TL_g\left(  \dot{ \xi } -\frac{1}{2}\operatorname{ad}^\dagger_\xi \xi -\frac{1}{2}\operatorname{ad}^\dagger_{\xi }\xi +\frac{1}{2}[\xi ,\xi ]\right) 
=TL_g\left(  \dot{ \xi } -\operatorname{ad}^\dagger_\xi \xi \right),
\]
where we used $X(g)= \dot{ g}$, $v_g= \dot{ g}$, so $f(g)= \dot{ g} g ^{-1} = \xi $ (respectively, $f(g) = g^{-1} \dot{g} = \xi$) and $ \mathbf{d} f( v_g)= \dot{ \xi }$.

Thus we obtain, due to the right-, or left-invariance of the metric $\gamma$,
  \begin{equation}
    L(g, \dot{g}, \ddot{g}) = \frac{1}{2} \left\|\frac{D}{Dt}\dot{g}\right\|_g^2 = \frac{1}{2}\left\|\dot{\xi} \pm \operatorname{ad}^\dagger_\xi\xi\right\|_{\mathfrak{g}}^2,
  \end{equation}
which depends only on the right invariant quantity
$\xi = \dot{g} g^{-1}$, respectively the left invariant
quantity $\xi = g^{-1} \dot{g}$.
 Accordingly, $L$ is right-, or left-invariant, and  the group-reduced Lagrangian is
\[
\ell( \xi , \dot{ \xi })= \frac{1}{2} \left\|\dot{ \xi } \pm\operatorname{ad}^\dagger_\xi \xi\right\|_{\mathfrak{g}}^2
\]
which completes the proof.
\end{proof}

\begin{remark} {\rm The above considerations generalize to geometric $k$-splines for $k > 2$. Indeed, iterated application of formulas \eqref{cov_left}, \eqref{cov_right} yields
\[
\frac{D^{k}}{Dt^{k}} \dot{g} = TR_g\left(\eta_{k}\right), \quad \text{respectively} \quad  \frac{D^{k}}{Dt^{k}} \dot{g} = TL_g\left(\eta_{k}\right),
\]
where the quantities $\eta_k \in \mathfrak{g}$ are defined by the recursive formulae
\begin{align}
\eta_1 = \dot{\xi} \pm\operatorname{ad}^\dagger_\xi \xi , \quad \text{and} \quad \eta_{k} = \dot{\eta}_{k-1} \pm \frac{1}{2} \left(\operatorname{ad}^\dagger_\xi \eta_{k-1} + \operatorname{ad}^\dagger_{\eta_{k-1}}\xi + \operatorname{ad}_{\eta_{k-1}}\xi\right),
\end{align}
for $\xi = \dot{g} g^{-1}$, respectively $\xi = g^{-1} \dot{g}$.
Therefore, the Lagrangian \eqref{k_splines_Lagr}  for geometric $k$-splines on a Lie group $G$ with right-, respectively left-invariant Riemannian metric,
\[
L_k\left( g,\dot g,..., g^{(k)} \right) = \frac{1}{2} \left\|\frac{D^{k-1}}{Dt^{k-1}}\dot g \right\|^2_g,
\]
is right-, respectively left-invariant, and the reduced Lagrangian is
\begin{equation}\label{Reduced_k_spline_Lagrangian}
\ell(\xi, \dot{\xi}, \ldots , \xi^{(k-1)}) = \frac{1}{2} \left\|\eta_{k-1}\right\|^2_{\mathfrak{g}}.
\end{equation}}
\end{remark}

\paragraph{Computing the second-order Euler-Poincar\'e equations for splines.}
Let us compute the Euler-Poincar\'e equations for $k=2$. The required variational derivatives of the reduced Lagrangian \eqref{reduced-lag-ell_dash} are given by
\begin{equation}
\label{var-der-lag-ell}
\frac{\delta \ell }{\delta  \dot{ \xi }}  =\dot{ \xi }^\flat \pm \operatorname{ad}^*_ \xi \xi^\flat=: \eta^\flat 
\quad\text{and}\quad  
\frac{\delta  \ell}{\delta \xi  }=\mp \left( \operatorname{ad}^*_{\eta} \xi ^\flat
+ \left(\operatorname{ad}_{\eta} \xi \right)^\flat \right) \in \mathfrak{g}  ^\ast .
\end{equation}
\begin{framed}\noindent
From formula \eqref{EP_2} with $k=2$ one then finds the $2^{nd}$-order Euler-Poincar\'e equation
\begin{equation}
\label{2nd-EPeqns1}
\left(\partial _t \pm \operatorname{ad}^*_ \xi \right) \left(\partial _t  \eta^\flat 
\pm  \operatorname{ad}^*_{ \eta} \xi ^\flat \pm \left(\operatorname{ad}_{\eta} \xi \right)^\flat\right) =0
, \quad \hbox{with}\quad
 \eta^\flat := \dot{ \xi }^\flat \pm \operatorname{ad}^* _ \xi \xi^\flat
,
\end{equation}
or, equivalently,
\begin{equation}
\label{2nd-EPeqns2}
\left(\partial _t \pm \operatorname{ad}^\dagger_ \xi \right) \left(\partial _t \eta
 \pm  \operatorname{ad}^\dagger_{ \eta} \xi  \pm \operatorname{ad}_{\eta} \xi \right) =0
 ,\quad \hbox{with}\quad
 \eta := \dot{ \xi } \pm \operatorname{ad}^\dagger _ \xi \xi.
\end{equation}
\end{framed}
These are the reduced equations for geometric $2$-splines associated to a left-, or right-invariant Riemannian metric on the Lie group $G$.

In an analogous fashion one can derive the Euler-Poincar\'{e} equations for geometric $k$-splines, using the reduced Lagrangian \eqref{Reduced_k_spline_Lagrangian}.

When the metric is left-, \emph{and} right-invariant (bi-invariant) further simplifications arise.

\paragraph{Example 1: Bi-invariant metric and the NHP equation.}
In the case of a \emph{bi-invariant} Riemannian metric, we have $ \operatorname{ ad}^\dagger _\xi \eta = - \operatorname{ ad}_ \xi \eta $ and therefore $ \eta^\flat= \dot{ \xi} ^\flat$, so that $ \eta= \dot{ \xi}$ and  the equations 
\eqref{2nd-EPeqns2} become
\begin{equation}
\left(\partial _t \pm \operatorname{ad}^*_ \xi \right) \ddot{\xi}^\flat =0\quad\text{or}\quad \left(\partial _t \pm \operatorname{ad}^\dagger_ \xi \right)  \ddot{\xi}=0\quad\text{or}\quad \dddot{\xi}\mp \left[\xi , \ddot{ \xi }\right]=0,
\label{CrSLe-commutator}
\end{equation}
as in \cite{CrSL1995}. Note that in this case, the reduced Lagrangian \eqref{reduced-lag-ell_dash} is simply given by $\ell( \xi , \dot \xi )=\frac{1}{2} \|\dot \xi \|^2$.
We also remark that since the metric is bi-invariant, one may choose to reduce the system either on the right \emph{or} on the left. This choice will determine which sign appears in \eqref{CrSLe-commutator}.

Taking $G=SO(3)$, we recover the NHP equation \eqref{Noakes-eqn} of \cite{NoHePa1989}:
\begin{equation}\label{NHP_equation}
\dddot{\mathbf{\Omega}} = \pm \mathbf{\Omega} \times \ddot{\mathbf{\Omega}}
\,.
\end{equation}
In \cite{NoHePa1989}, the unreduced equations in the general case are also derived, but the symmetry reduced equation is given only for $SO(3)$ with bi-invariant metric.

\begin{remark}{\rm [Conventions for $\mathfrak{so}(3)$ and $\mathfrak{so}(3)^*$]\newline
In equation \eqref{NHP_equation} and throughout the paper we use vector notation for the Lie algebra $\mathfrak{so}(3)$ of the Lie group of rotations $SO(3)$, as well as for its dual $\mathfrak{so}(3)^*$. One identifies $\mathfrak{so}(3)$ with $\mathbb{R}^3$ via the familiar isomorphism
\begin{equation}\label{hat_map}
\,\widehat{\,}:  \mathbb{R}^3\rightarrow \mathfrak{so}(3),\quad\mathbf{\Omega}=\left(
\begin{array}{c}
a\\
b\\
c\\
\end{array}\right) \mapsto \Omega:=\widehat{\mathbf{\Omega}}=\left(
\begin{array}{ccc}
0&-a&b\\
a&0&-c\\
-b&c&0

\end{array}
\right),
\end{equation}
called the \emph{hat map}.
This is a Lie algebra isomorphism when the vector cross product $\times$ is used as the Lie bracket operation on 
$\mathbb{R}^3$. The identification of $\mathfrak{so}(3)$ with $\mathbb{R}^3$ induces an isomorphism of the dual spaces $\mathfrak{so}(3)^* \cong \left(\mathbb{R}^3\right)^* \cong \mathbb{R}^3$.}
\end{remark}

\paragraph{Example 2: Elastica.}
Another example of the $2^{nd}$-order Euler-Poincar\'e equation arises in the case of \emph{elastica} treated in \cite{HuBl2004a}, whose Lagrangian is
\[
L(g, \dot{ g}, \ddot{ g})=\frac{\tau^{2} }{2} \| \dot{ g}\|_g^2+ \frac{1}{2} \left \| \frac{D}{Dt}\dot g\right\|_g^2,
\]
and whose reduced Lagrangian is
\begin{equation}\label{red_lagr_elastica}
\ell( \xi , \dot{ \xi })= \frac{\tau^{2}  }{2} \| \xi \|_\mathfrak{g} ^2 +\frac{1}{2} \|\dot{ \xi } \pm\operatorname{ad}^\dagger_\xi \xi\|_\mathfrak{g}^2.
\end{equation}
Using the $2^{nd}$-order Euler-Poincar\'e equation \eqref{EP_2} one easily obtains the reduced equations
\begin{equation}
\left(\partial _t \pm \operatorname{ad}^\dagger_ \xi \right) \left(\partial _t \eta
 \pm  \operatorname{ad}^\dagger_{ \eta} \xi  \pm \operatorname{ad}_{\eta} \xi - \tau^{2} \xi \right) =0
 ,\quad \hbox{with}\quad
 \eta := \dot{ \xi } \pm \operatorname{ad}^\dagger _ \xi \xi,
\end{equation}
which simplify to
\[
\left(\partial _t \pm \operatorname{ad}^\dagger_ \xi \right) \left(\partial _t^2 \xi- \tau^{2} \xi \right) =0
\]
in the bi-invariant case.

\begin{remark}{\rm We now consider the particular case $G=SO(3)$. 
Let $\mathbf{I}$ be a $3\times 3$ symmetric positive definite matrix (inertia tensor) and consider the inner product $\gamma(\mathbf{\Omega}_1,\mathbf{\Omega}_2)= \mathbf{I}\mathbf{\Omega}_1\cdot \mathbf{\Omega}_2$ on $\mathbb{R}^3$.
The Lagrangian for the elastica on $SO(3)$ reads
\[
L( \Lambda  , \dot \Lambda  , \ddot \Lambda  )
=
\frac{\tau^{2}}{2}\left\| \dot \Lambda \right\|_{\Lambda}^2 
+ 
\frac{1}{2} \left\| \frac{D}{Dt} \dot \Lambda\right\|_{\Lambda}^2
,\]
where $ \|\cdot\|_{\Lambda}$ is the right-, respectively 
left-invariant metric induced the inner product $\gamma$. Relative to this inner product we have 
\[
\operatorname{ad}^\dagger_{\mathbf{\Omega}_1}\mathbf{\Omega}_2=  \mathbf{I}  ^{-1} ( \mathbf{I}\,\mathbf{\Omega}_2 \times \mathbf{\Omega }_1)
,
\] 
so the reduced Lagrangian \eqref{red_lagr_elastica} reads
\begin{align} 
\ell (\boldsymbol{\Omega}, \dot{ \boldsymbol{\Omega}})
&= 
\frac{\tau^{2}}{2} \| \boldsymbol{\Omega} \| ^2
+ 
\frac{1}{2} \| \boldsymbol{\dot{\Omega}}\pm \mathbf{I} ^{-1} ( \mathbf{I}  \boldsymbol{\Omega} \times \boldsymbol{\Omega}) \| ^2
\nonumber\\
&=
\frac{\tau^{2}}{2} \boldsymbol{\Omega} \cdot \mathbf{I} \boldsymbol{\Omega}
+ \frac{1}{2}  \left( \mathbf{I} \boldsymbol{\dot{\Omega}}\pm   \mathbf{I} \boldsymbol{\Omega} \times \boldsymbol{\Omega}\right)  \cdot \mathbf{I}^{-1}    \left( \mathbf{I} \boldsymbol{\dot{\Omega}}\pm   \mathbf{I}   \boldsymbol{\Omega} \times \boldsymbol{\Omega}\right).
\label{reduced-lag}
\end{align}
If $\tau=0$, this expression can be interpreted as the Lagrangian for geometric $2$-splines of a rigid body.

If $\mathbf{I} $ is the identity, the Lagrangian in (\ref{reduced-lag})simplifies to
\[
\ell(\boldsymbol{\Omega}, \dot{ \boldsymbol{\Omega}})=\frac{\tau^{2}}{2} \boldsymbol{\Omega} \cdot \boldsymbol{\Omega}
+ \frac{1}{2}  \boldsymbol{\dot{\Omega}} \cdot  \boldsymbol{\dot{\Omega}}
\]
and the Lagrangian of the NHP equation is recovered when $\tau=0$.
}
\end{remark}

\paragraph{Example 3: $L^2$-splines.} One can consider $L^2$ geometric $2$-splines on the diffeomorphism group of a manifold $ \mathcal{D} $ as follows. Fix a Riemannian metric $g$ on $ \mathcal{D} $ and consider the associated $L^2$ right-invariant Riemannian metric on $G= \operatorname{Diff}( \mathcal{D} )$ and its induced second-order Lagrangian
\[
L( \eta , \dot \eta , \ddot \eta )= \frac{1}{2} \left\| \frac{D}{Dt} \dot \eta \right\|_\eta^2
\]
on $T^{(2)} \operatorname{Diff}( \mathcal{D} )$. The reduced Lagrangian on $ 2 \mathfrak{g}  = 2 \mathfrak{X}  ( \mathcal{D} )$ reads
\[
\ell( \mathbf{u} , \dot{\mathbf{u}})= \frac{1}{2} \| \dot{\mathbf{u}} + \operatorname{ad}^\dagger_ \mathbf{u} \mathbf{u} \|^2,
\]
where $\operatorname{ad}^\dagger$ denotes the transpose with respect to the $L^2$ inner product, given by  $\operatorname{ad}^\dagger_{\mathbf{u}} \mathbf{v} = \nabla_\mathbf{u} \mathbf{v} + \left(\nabla \mathbf{u}\right)^\mathsf{T} \cdot \mathbf{v} + \mathbf{v} \operatorname{div} \mathbf{u} $. In this case, the $ \operatorname{ad}^\dagger$ and $ \operatorname{ad}$ terms  in \eqref{2nd-EPeqns2} combine to produce the spline equation
\[
\left( \partial _t + \operatorname{ad}_ {\mathbf{u}} ^\dagger \right) \left( \partial _t \mathbf{v} +2\mathsf{S}\mathbf{v}\cdot \mathbf{u} + \mathbf{u} \operatorname{div} \mathbf{v} \right) =0,\quad \mathbf{v} = \partial _t \mathbf{u} + 2\mathsf{S}\mathbf{u}\cdot \mathbf{u}+ \mathbf{u} \operatorname{div} \mathbf{u},
\]
where $ \mathsf{S}\mathbf{u} := \left( \nabla \mathbf{u} + \nabla \mathbf{u}^\mathsf{T} \right) /2$ is the strain-rate tensor.

In the incompressible case, that is when $G= \operatorname{Diff}_{vol}( \mathcal{D} )$, the transpose of $\operatorname{ad}$ relative to the $L^2$ inner product on divergence free vector field is denoted by $\operatorname{ad}^+$ and is related to $\operatorname{ad}^\dagger$ by the formula
\[
\operatorname{ad}^+_ \mathbf{u} \mathbf{v} =\mathbb{P}\left(\operatorname{ad}^\dagger_ \mathbf{u} \mathbf{v}\right)=\mathbb{P}\left( \nabla _ \mathbf{u}\mathbf{v} + \left(\nabla \mathbf{u}\right)^\mathsf{T} \cdot \mathbf{v} \right),
\]
where $ \mathbb{P}$ denotes the Hodge projector onto the divergence free vector fields. In this case \eqref{2nd-EPeqns2} reads
\[
\left( \partial _t + \operatorname{ad}^+_\mathbf{u} \right) \left( \partial _t \mathbf{v} +\mathbb{P}\left( \nabla_ \mathbf{v}  \mathbf{u}+ \left(\nabla \mathbf{v}\right)^\mathsf{T} \cdot \mathbf{u} \right)+ \nabla _ \mathbf{u} \mathbf{v} - \nabla _ \mathbf{v} \mathbf{u}  \right) =0,\quad \mathbf{v} = \partial _t \mathbf{u} +2 \mathbb{P}\left( \mathsf{S} \mathbf{u} \cdot \mathbf{u} \right),\quad \operatorname{div} \mathbf{u} =0.
\]
Remarkably, using the formula $\operatorname{ ad}^\dagger _ \mathbf{u} \nabla p= \nabla ( \nabla p \cdot \mathbf{u} )$ for $\operatorname{div} \mathbf{u} =0$, all the gradient terms arising from the Hodge projector can be assembled in a single gradient term, thereby producing the incompressible $2$-spline equations
\[
\left( \partial _t + \operatorname{ad}^\dagger_\mathbf{u} \right) \left( \partial _t \mathbf{w} +2\mathsf{S} \mathbf{w} \cdot \mathbf{u}  \right)=- \nabla p,\quad \mathbf{w} = \partial _t \mathbf{u} +2  \mathsf{S} \mathbf{u} \cdot \mathbf{u},\quad \operatorname{div}\mathbf{u} =0,
\]
where $ \operatorname{ad}^\dagger$ (and not $\operatorname{ad}^+$) is used.

\paragraph{Example 4: $H^1$-splines.} One can alternatively consider splines relative to the right-invariant metric induced by the $H^1$ inner product $\left\langle Q\mathbf{u} , \mathbf{u} \right\rangle $, where $Q=(1- \alpha ^2 \Delta ) $. In this case, the $2$-spline equation reads
\begin{equation}\label{H1_splines}
\left( \partial _t + \operatorname{ad}^\dagger_ \mathbf{u} \right) \left(\partial _t Q\mathbf{v}+ 2 \mathsf{S}( Q \cdot \operatorname{ad}_ \mathbf{v} ) \cdot \mathbf{u} \right) =0,\quad Q \mathbf{v} = \partial _t Q \mathbf{u} + \operatorname{ad}^\dagger_ \mathbf{u} Q \mathbf{u} 
\end{equation}
where $\mathsf{S}(L):= \frac{1}{2} (L+L^*)$.

\begin{remark}
{\rm Note that in order to obtain the simple expression
\[
\ell( \xi , \dot \xi )= \frac{1}{2} \| \xi \|_1^2+ \frac{1}{2} \| \dot\xi\|_2 ^2
\]
(instead of \eqref{red_lagr_elastica}) for $\ell$ by reduction, one needs to modify the spline Lagrangian as
\[
L( g, \dot g, \ddot g)= \frac{1}{2} \left\| \dot g\right\|_1^2 + \frac{1}{2} \left\| \frac{D}{Dt} \dot g \pm \operatorname{ad}^\dagger_{\dot g}\dot g\right\| ^2_2 ,
\]
where $\|\cdot\|_i$, $i=1,2$, are two norms associated to two $G$-invariant Riemannian metrics on $G$ and
$\operatorname{ad}^\dagger$ is extended as a bilinear map $T_gG \times T_gG \rightarrow T_gG$ by $G$-invariance.

The associated $2^{nd}$ order Euler-Poincar\'e equations are simpler than the one associated to $2$-splines. For example, the reduced Lagrangian $\ell( \mathbf{u} , \dot{\mathbf{u}} )= \frac{1}{2} \left\langle ( 1- \alpha ^2 \Delta ) \mathbf{u} , \mathbf{u} \right\rangle + \frac{1}{2} \| \dot{ \mathbf{u} }\|^2 $ produces the following modification of the EPDiff equation:
\begin{equation}\label{modified_H1_splines}
\left( \partial _t + \operatorname{ad}^\dagger_ \mathbf{u} \right) \left(  \mathbf{u} - ( \alpha ^2 \Delta \mathbf{u} + \mathbf{u} _{tt}) \right) =0.
\end{equation}
}
\end{remark}

\subsection{Parameter dependent Lagrangians}\label{subsec_param_dep_Lagr}
In many situations, such as the heavy top or the compressible fluid, the Lagrangian of the system is defined on the tangent bundle $TG$ of the configuration Lie group $G$, but it is not $G$-invariant. In these cases, the Lagrangian depends parametrically on a quantity $ q _0$ in a manifold $Q$ on which $G$ acts and that breaks the symmetry of the Lagrangian $L=L_{ q _0 }$. We refer to \cite{HoMaRa1998}, \cite{GBRa2009} for the case of (affine) representation on vector spaces, relation with semidirect products and many examples. This theory was extended to arbitrary actions on manifolds in \cite{GBTr2010} for applications to symmetry breaking phenomena. We now briefly present the extension of this theory to the case of higher order Lagrangians.

Consider a $k^{th}$-order Lagrangian $L_{q_0 } : T^{(k)}G \rightarrow \mathbb{R}  $, depending on a parameter $ q _0$ in a manifold $Q$. We suppose that $G$ acts on the manifold $Q$ and that the Lagrangian $L$ is $G$-invariant under the action of $G$ on both $T^{(k)}G$ and $Q$, where we now see $L$ as a function defined on $T^{(k)}G \times Q$. Concerning the action of $G$ on $T^{(k)}G \times Q$, there are several variants that one needs to consider, since they all appear in applications. 

\medskip

$(1)$ First, one has the right, respectively the left action
\begin{eqnarray*}
(g, \dot g, ..., g^{(k)}, q _0 ) &\mapsto& \left( gh, \dot gh, ..., g^{(k)}h, h^{-1}q _0 \right)
,\\
\text{respectively}\quad(g, \dot g, ..., g^{(k)}, q _0 ) &\mapsto& \left( hg, h\dot g, ...,h g^{(k)}, q _0 h^{-1}\right).
\end{eqnarray*}
The reduced variables are $( \xi , q) =( \dot g g ^{-1}, g q _0  )$, respectively $( \xi , q) =( g ^{-1} \dot g , q _0 g )$. In this case, the $k^{th}$-order Euler-Lagrange equations for $L_{ q _0 }$ on $T^{(k)}G$ (where $ q _0 \in Q$ is a fixed parameter) are equivalent to the $k^{th}$-order Euler-Poincar\'e equations together with the advection equation
\begin{framed}\begin{equation}\label{EP_with_q_RL}
\left(\partial _t  \pm \operatorname{ad}^*_ \xi \right)  
\left(
\sum _{ j=0}^ {k-1} (-1)^j \partial _t ^j  \frac{\delta \ell }{\delta  \xi^{ (j) } }
\right)
=
\mathbf{J} \left( \frac{\delta   \ell}{ \delta  q} \right),\quad \partial _t q - \xi _Q (q)=0,
\end{equation}
\end{framed}\noindent
with initial condition $q _0$. Here $ \xi _Q(q)$ denotes the infinitesimal generator of the $G$ action on $Q$ and $ \mathbf{J} : T^*Q \rightarrow \mathfrak{g}  ^\ast $ defined by $ \left\langle \mathbf{J} ( \alpha _q ) , \xi \right\rangle := \left\langle \alpha _q , \xi _Q (q) \right\rangle $ denotes the momentum map associated to the $G$ action on $T^*Q$.

The associated variational principle reads
\[
\delta \int_{t _1 }^{t _2 }\ell\left(  \xi , \dot \xi , ..., \xi ^{(k-1)}, q \right)dt =0
,\]
relative to the constrained variations \eqref{constr_var} and constrained variations of $q$ given by $ \delta q =\eta  _Q(q)$, where $ \eta = (\delta g )g ^{-1} $, respectively
$\eta = g^{-1}(\delta g)$. Equations \eqref{EP_with_q_RL} and their variational formulation can be obtained by an easy generalization of the approach used in Section \S\ref{k_order_EP_sec}.

\medskip

$(2)$ Secondly, one can consider the right, respectively the left action
\begin{eqnarray*}
(g, \dot g, ..., g^{(k)}, q _0 ) &\mapsto& \left( gh, \dot gh, ..., g^{(k)}h, q _0h \right) 
,\\
\text{respectively}\quad(g, \dot g, ..., g^{(k)}, q _0 ) &\mapsto& \left( hg, h\dot g, ...,h g^{(k)}, hq _0\right).
\end{eqnarray*}
The reduced variables are $( \xi , q) =( \dot g g ^{-1},  q _0g ^{-1}   )$, respectively $( \xi , q) =( g ^{-1} \dot g ,g ^{-1}  q _0 )$ and one gets the reduced equations
\begin{framed}\begin{equation}\label{EP_with_q_RR}
\left(\partial _t  \pm \operatorname{ad}^*_ \xi \right)  
\left(\sum _{ j=0}^ {k-1} (-1)^j \partial _t ^j  \frac{\delta \ell }{\delta  \xi^{ (j) } }\right)
=
-\mathbf{J} \left( \frac{\delta   \ell}{ \delta  q} \right),
\qquad \partial _t q + \xi _Q (q) =0,
\end{equation}
\end{framed}\noindent
with initial condition $q _0$.

\medskip

If $Q=V^*$ is the dual of a $G$-representation space so
which $G$ acts on $V ^\ast$ by the dual representation, the above equations reduce to 
\begin{equation}
\left(\partial  _t \pm \operatorname{ad}^*_ \xi \right)  
\left(\sum _{ j=0}^ {k-1} (-1)^j \partial _t ^j  \frac{\delta \ell }{\delta  \xi^{ (j) } }
\right)
=\frac{\delta \ell}{\delta  a}\diamond a,
\end{equation}
where the diamond operation $\diamond: V \times V ^\ast \rightarrow \mathfrak{g}  ^\ast $ is defined by 
\[
 \left\langle v \diamond a, \xi \right\rangle = \left\langle a, \xi _V(v) \right\rangle,
 \]
 and therefore $ \mathbf{J} (a,v)=- v \diamond a$. These are the higher order version of the Euler-Poincar\'e equations with advected quantities studied in \cite{HoMaRa1998}.

\paragraph{Example : Rate-dependent fluid models.}
Rate-dependent fluid models are usually defined using Lagrangians that depend on the strain-rate tensor ${\sf S}:=(\nabla\mathbf{u}+ (\nabla\mathbf{u})^\mathsf{T})/2$ and its higher \emph{spatial} derivatives \cite{BeFoHoLi1988}. A related class of spatially-regularized fluid models have been introduced as turbulence models \cite{FoHoTi2001}.  

Yet another class of rate-dependent fluid models may be defined, e.g., as $2nd$ order Euler-Lagrange equations $T^{(2)} \operatorname{Diff}( \mathcal{D} )$ for a parameter dependent Lagrangian $L_{a_0}$. The group reduced representation of the equations of motion for such rate-dependent fluids is found from the previous manipulations,
namely
\begin{equation}\label{rate_dep_fluid}
\left(\partial _t \pm \pounds _\mathbf{u}  \right) 
\left( \frac{\delta  \ell}{\delta \mathbf{u}  } - \partial _t \frac{\delta \ell }{\delta \dot{ \mathbf{u} } } \right) 
= \frac{\delta \ell }{\delta  a} \diamond a
,\qquad 
\left( \partial _t + \pounds _ \mathbf{u} \right) a=0.
\end{equation}
One has the $+$ sign for right and the $-$ sign for left invariance. The usual Eulerian fluid representation is right-invariant and so takes the $+$ sign. Physically, these fluid models penalize the flow for producing higher temporal frequencies. Therefore, these models might be considered as candidates for \emph{frequency-regularized} models for fluid turbulence. The Kelvin theorem for these fluids involves circulation of the higher time derivatives. For right-invariant higher-order Lagrangians, the Kelvin theorem becomes
\[
\frac{d}{dt}\oint_{c(\mathbf{u} )}
\frac{1}{D}
\left( \frac{\delta  \ell}{\delta \mathbf{u}  } - \partial _t \frac{\delta \ell }{\delta \dot{ \mathbf{u} } } \right) 
=
\oint_{c(\mathbf{u} )}
\frac{1}{D}
\frac{\delta \ell }{\delta  a} \diamond a
,\]
where the density $D$ satisfies the continuity equation $( \partial _t + \pounds _ \mathbf{u}) D=0$. Consequently, the integrands in the previous formula are 1-forms and thus may be integrated around the closed curve $c(\mathbf{u} )$ moving with the fluid velocity, $\mathbf{u}$. This is the statement of the  
Kelvin-Noether theorem \cite{HoMaRa1998} for $k$-splines.

\subsection{Splines with constraints}\label{subsec_constraints}
Suppose that one wants to minimize the action $\int_{t _0 }^{ t _1 }L(q, \dot q, \ddot q)dt $ over curves $q(t)\in Q$ subject to the condition
\[
\left\langle \omega _i (q), \dot q \right\rangle =k _i ,\;\; i=1,...,k,
\]
where $ \omega _i $ are $1$-forms and $ k _i \in \mathbb{R}  $. One uses the variational principle
\begin{equation}\label{variational_principle_constraints}
\delta \int_{t _0 }^{ t _1 }\left( L(q, \dot q, \ddot q)+ \sum_{i=1}^k\lambda _i \left( \left\langle \omega _i (q), \dot q \right\rangle- k _i \right)  \right) dt=0,
\end{equation}
for arbitrary variations $\delta\lambda_i$ of the curves $\lambda_i(t)$, $i=1,...,k$, and for variations $\delta q$ vanishing at the endpoints. Variations relative to $q$ yield the equation
\[
\frac{d^2 }{dt^2 } \frac{\partial  L}{\partial  \ddot q}-\frac{d}{dt} \frac{\partial  L}{\partial  \dot q}+ \frac{\partial  L}{\partial  q} = \sum_{i=1}^k\left(  \lambda _i \mathbf{i} _{ \dot q} \mathbf{d} \omega _i + \dot \lambda _i \omega _i \right),
\]
whereas variations relative to $\lambda_i$ yield the constraint.
For example, for the Lagrangian \eqref{2_splines_Lagr} this yields the equations
\[
\frac{D^3}{Dt^3}\dot q(t)+ R \left( \frac{D}{Dt}\dot q(t),\dot q(t) \right) \dot q(t)= \tau ^2 \frac{D}{Dt}\dot q (t)+\sum_{i=1}^k\left(  \lambda _i \mathbf{i} _{ \dot q} \mathbf{d} \omega _i + \dot \lambda _i \omega _i \right),
\]
as in \cite{BlCr1993}, \cite{CrSL1995}, \cite{HuBl2004}, where $ \omega _i = X _i ^\flat$ for given linearly independent vector fields $X_1, \ldots X_k \in 
\mathfrak{X}(Q)$.

\begin{remark}\rm
The variational principle \eqref{variational_principle_constraints} is equivalent to
\begin{equation}\label{variational_principle_constraints_simplified}
\delta \int_{t _0 }^{ t _1 }\left( 
L(q, \dot q, \ddot q)
+ 
\sum_{i=1}^k\lambda _i  \left\langle \omega _i (q), \dot q\, \right\rangle \right) dt
=0
\quad\text{and}\quad \left\langle \omega _i (q), \dot q \right\rangle =k _i ,\;\; i=1,...,k,
\end{equation}
where only variations of $q$ are involved and the term containing $k_i$ is suppressed.
\end{remark}

We now consider the special case $Q=G$ and we suppose that the one-forms $ \omega _i$, $i=1,...,k$, are $G$-invariant. That is, we can write
\[
\left\langle \omega_i (g),v _g \right\rangle = \left\langle \omega _i (g) g^{-1} , v _g g ^{-1} \right\rangle = \left\langle \zeta_i  , v _g g ^{-1}  \right\rangle \quad\text{resp.}\quad \left\langle \omega_i(g),v _g \right\rangle= \left\langle \zeta_i  ,  g ^{-1}v _g \right\rangle,
\]
where
\[
\zeta _i := \omega _i (e)\in \mathfrak{g}  ^\ast .
\]
The reduction of the variational principle \eqref{variational_principle_constraints} yields the constrained variational principle
\[
\delta \int_{t _0 }^{ t _1 } \left( \ell( \xi , \dot \xi ) + \sum_{i=1}^k \lambda _i \left( \left\langle \zeta _i , \xi \right\rangle - k _i \right) \right) dt=0,
\]
for arbitrary variations $\delta\lambda_i$ of the curves $\lambda_i(t)$, $i=1,...,k$, and variations of $\xi(t)$ satisfying the constraints \eqref{constr_var}. Equivalently, using \eqref{variational_principle_constraints_simplified} we rewrite the stationarity condition as
\[
\delta \int_{t _0 }^{ t _1 } \left( \ell( \xi , \dot \xi ) + \left\langle z, \xi \right\rangle \right) dt=0\quad\text{and}\quad \left\langle \zeta _i ,\xi  \right\rangle =k _i,\;\; i=1,...,k,
\]
for variations of $\xi$ satisfying \eqref{constr_var} and where we have defined $z:= \sum_{i=1}^k\lambda _i\zeta _i\in \mathfrak{g}  ^\ast $. We obtain the equations
\begin{framed}{
\[
\left( \partial _t \pm \operatorname{ad}^*_ \xi \right) \left( \frac{\delta  \ell}{\delta \xi  } - \partial _t \frac{\delta  \ell}{\delta \dot \xi  }+ z \right) =0,\quad \left\langle \zeta _i ,\xi  \right\rangle =k _i .
\]}
\end{framed}
With the Lagrangian \eqref{reduced-lag-ell_dash} for $2$-splines, we find the reduced equations
\[
\left(\partial _t \pm \operatorname{ad}^\dagger_ \xi \right) \left(\partial _t \eta
\pm  \operatorname{ad}^\dagger_{ \eta} \xi  \pm \operatorname{ad}_{\eta} \xi -z\right) =0,\quad \hbox{with}\quad\eta := \dot{ \xi } \pm \operatorname{ad}^\dagger _ \xi \xi
\]
and for bi-invariant metrics, we get
\[
\left(\partial _t \pm \operatorname{ad}^\dagger_ \xi \right) \left(\ddot \xi 
-z\right) =0,
\quad\text{i.e.,}\quad 
\dddot \xi \mp [ \xi , \ddot \xi -z]- \dot z=0, \quad \left\langle \zeta _i ,\xi  \right\rangle =k _i,
\]
which coincides with equation (39) in \cite{CrSL1995}. See also \cite{BlCr1996a}.

We can also consider higher-order constraints, with associated variational principle
\[
\delta \int_{t _0 }^{ t _1 } \left( \ell( \xi , \dot \xi , ..., \xi ^{( k-1)})  + \left\langle z _0 , \xi \right\rangle +...+ \left\langle z _{k-1}, \xi ^{ (k-1) }\right\rangle \right) dt= 0.
\]
In this case, one obtains the equations
\[
\left(\partial _t  \pm \operatorname{ad}^*_ \xi \right)  
\sum _{ j=0}^ {k-1} (-1)^j \partial _t ^j \left(  \frac{\delta \ell }{\delta  \xi^{ (j) } }+ z _j \right) =0.
\]
For example, with $k=2$, and a bi-invariant metric we have
\[
\left(\partial _t \pm \operatorname{ad}^\dagger_ \xi \right) \left(\ddot \xi + \dot z _1 
-z_0 \right) =0.
\]

\section{Clebsch-Pontryagin optimal control}\label{Clebsch-sec}

Here we develop the $k^{th}$-order Euler-Poincar\'e 
equations from an optimal control approach. The ideas 
in \cite{GBRa2010} for $k=1$ generalize easily to higher order.

\begin{definition}\label{def_Clebsch} Let $\Phi$ be a (right or left) action of a Lie group $G$ on a manifold $Q$. For a Lie algebra element $\xi \in\mathfrak{g}$ let
\[
\left.\xi _Q(q):=\frac{d}{dt}\right|_{t=0}\Phi_{\operatorname{exp}(t \xi )}(q)
,\]
denote the corresponding infinitesimal generator of the action. Given a cost function $\ell:k\mathfrak{g}\rightarrow\mathbb{R}$, the \textbf{Clebsch-Pontryagin optimal control problem} is, by definition,
\begin{equation}\label{Clebsch_optimal_control}
\min_{\xi (t) } \int_0^T\ell \left( \xi , \dot \xi ,....,  \xi ^{ (k-1) }\right) dt
\end{equation}
subject to the following conditions:
\begin{itemize}
\item[\rm (A)]\ $\dot q=\xi _Q(q)\quad$ or $\quad({\rm A})'$\ \ $\dot
q=-\xi _Q(q)$;
\item[\rm (B)]\ $q(0)=q_0$ \quad  and \quad  $q(T)=q_T$;
\item[\rm (C)]\ $\xi^{(j)}(0)=\xi^j_0 \quad$ and $\quad \xi^{(j)}(T)=\xi^j_T,\quad j=0,...,k-2$,
\end{itemize}
where $q_0, q_T\in Q$ and $\xi^j_0, \xi^j_T\in\mathfrak{g}$,  $j=0,...,k-2$, are given.
\end{definition}

\paragraph{Variational equations.} We suppose that condition (A) of Definition \ref{def_Clebsch} holds.  (The calculation for case ${\rm (A)}'$ is similar.) The resolution of this problem uses the Pontryagin maximum principle which, under sufficient smoothness conditions, implies that its solution necessarily satisfies the variational principle
\[
\delta \int _{0}^{T}\left(  \ell\left( \xi , \dot \xi ,...., \xi ^{ (k-1) }\right)+\langle\alpha,\dot{q}-\xi  _Q (q)\rangle \right) dt=0,
\]
for curves $ t \mapsto \xi (t)\in \mathfrak{g}  $ and $ t \mapsto \alpha (t)\in T_{q(t)}^*Q$. This variational principle 
yields the conditions
\begin{equation}\label{stationarity_conditions_Clebsch}
\mathbf{J} ( \alpha  (t))=\sum _{ j=0}^ {k-1} (-1)^j \partial _t ^j\frac{\delta \ell }{\delta  \xi^{ (j) } }\quad\text{and}\quad \dot\alpha=\xi _{T^*Q}(\alpha),
\end{equation}
in which $ \mathbf{J} : T^\ast Q\rightarrow \mathfrak{g}^* $ is the cotangent bundle momentum map, as above, and $ \xi _{T^*Q}$ denotes the infinitesimal generator of the cotangent lifted action, denoted 
$\Phi^{T^*}:G \times T^*Q \rightarrow T^*Q$.

If $G$ acts on the right (respectively left), a solution $\alpha(t)$ of $\dot\alpha=\xi _{T^*Q}(\alpha)$ is necessarily of the form $\alpha(t)=\Phi_{g(t)}^{T^*}(\alpha(0))$, where $g(0)=e$ and $\dot g(t)g(t)^{-1}=\xi (t)$ (respectively 
$g(t)^{-1}\dot g(t)=\xi (t)$). The above conditions imply coadjoint motion,
\[
\mathbf{J}(\alpha(t))=\mathbf{J}\left(\Phi^{T^*}_{g(t)}(\alpha(0))\right)=\operatorname{Ad}_{g(t)}^*\mathbf{J}(\alpha(0)),\quad\text{respectively}\quad \mathbf{J}(\alpha(t))=\operatorname{Ad}_{g(t)^{-1} }^*\mathbf{J}(\alpha(0))
,
\]
and by differentiating relative to time, we obtain the left (right) Euler-Poincar\'e equations:
\[
\frac{d}{dt}\mathbf{J}(\alpha(t))=\operatorname{ad}^*_{g(t)^{-1}\dot{g}(t)}\mathbf{J}(\alpha(t))=\operatorname{ad}^*_{\xi (t)}\mathbf{J}(\alpha(t)),
\]
respectively
\[
\frac{d}{dt}\mathbf{J}(\alpha(t))=-\operatorname{ad}^*_{\dot g(t)g(t)^{-1}}\mathbf{J}(\alpha(t))=-\operatorname{ad}^*_{\xi (t)}\mathbf{J}(\alpha(t)).
\]
Upon using the first condition in \eqref{stationarity_conditions_Clebsch}, we recover the $k^{th}$-order Euler-Poincar\'e equations,
\begin{equation}\label{EP_clebsch}
\left(\partial _t \mp \operatorname{ad}^*_ \xi \right)  \sum _{ j=0}^ {k-1} (-1)^j \partial _t ^j\frac{\delta \ell }{\delta  \xi^{ (j) } }=0.
\end{equation}

\paragraph{Example 1: Clebsch approach to the NHP equations.} The NHP equations can be obtained from the Clebsch approach by considering the action of $SO(3)$ on $ \mathbb{R}  ^3 $. The Clebsch-Pontryagin control problem is
\[
\min_{ \xi (t)} \int_0^T \| \dot{ \mathbf{\Omega} }\| ^2 dt,\quad\text{subject to}\quad   \dot {\mathbf{q}}= \mathbf{\Omega} \times \mathbf{q},\;\; \mathbf{q}(0)=\mathbf{q} _0 ,\;\; \mathbf{q}(T)= \mathbf{q}_T,\;\;\mathbf{\Omega}(0)=\mathbf{\Omega}_0,\;\;\mathbf{\Omega}(T)=\mathbf{\Omega}_T.
\]
The stationarity conditions \eqref{stationarity_conditions_Clebsch} read $ \mathbf{q} \times \mathbf{p}=- \ddot{\mathbf{\Omega}} $, $ \dot {\mathbf{q}}= \mathbf{\Omega} \times \mathbf{q}$ and $\dot {\mathbf{p}}= \mathbf{\Omega} \times \mathbf{p}$. One directly observes that they imply the NHP equations.

\paragraph{Example 2: Clebsch approach to $H^1$-splines.} We let the diffeomorphism group $ \operatorname{Diff}_{vol}( \mathcal{D} )$ act on the left on the space of embeddings 
$\operatorname{Emb}(S,\mathcal{D})$ of a manifold $S$ in
$D$. The associated Clebsch-Pontryagin control problem is
\begin{align*}
\min_{ \mathbf{u} (t)} \int_0^T \left\| \dot{\mathbf{u}} + \operatorname{ad}_ \mathbf{u} ^\dagger \mathbf{u} \right\|^2_{H^1} dt,\quad\text{subject to}\quad &\dot{ \mathbf{Q} }= \mathbf{u} \circ \mathbf{Q} ,\;\; \mathbf{Q} (0)= \mathbf{Q} _0 ,\;\; \mathbf{Q} (T)= \mathbf{Q} _T,\\
&\mathbf{u}(0)=\mathbf{u}_0,\;\;\mathbf{u}(T)=\mathbf{u}_T.
\end{align*}
The condition
\[
\mathbf{J} ( \mathbf{Q} , \mathbf{P} )=\partial _t Q\mathbf{v}+ 2 \mathsf{S}( Q \cdot \operatorname{ad}_ \mathbf{v} ) \cdot \mathbf{u} ,
\]
where $Q \mathbf{v} = \partial _t Q \mathbf{u} + \operatorname{ad}^\dagger_ \mathbf{u} Q \mathbf{u} $, together with the Hamilton equations on $T^* \operatorname{Emb}(S, \mathcal{D} )$ imply the $H^1$ spline equations \eqref{H1_splines}. In the case of equations \eqref{modified_H1_splines} the condition \eqref{stationarity_conditions_Clebsch} reads
\[
\mathbf{J}  ( \mathbf{Q} , \mathbf{P} )= Q \mathbf{u} - \mathbf{u} _{tt}.
\]

\paragraph{Additional $q$-dependence in the Lagrangian.} 
One can easily include a $q$-dependence in the cost function $\ell$ of the Clebsch-Pontryagin optimal control problem \eqref{Clebsch_optimal_control}. In this case, the stationarity conditions \eqref{stationarity_conditions_Clebsch} become
\begin{equation}\label{stationarity_conditions_Clebsch_q}
\mathbf{J} ( \alpha  (t))=\sum _{ j=0}^ {k-1} (-1)^j \partial _t ^j\frac{\delta \ell }{\delta  \xi^{ (j) } }\quad\text{and}\quad \dot\alpha=\xi _{T^*Q}(\alpha)+ \operatorname{Ver}_\alpha \frac{\delta  \ell}{\delta  q},
\end{equation}
where  for $\alpha, \beta \in T ^\ast _qQ$, the \textit{vertical lift of $\beta $ relative to $\alpha$} is defined by
\[
\operatorname{Ver}_ \alpha \beta : = \left.\frac{d}{ds}\right|_{s=0} ( \alpha+ s \beta) \in T_{ \alpha} ( T ^\ast Q).
\]
In this case, the differential equation \eqref{EP_clebsch} for the control $ \xi (t) $ generalizes to 
\begin{equation}\label{q-dep-eqn}
\left(\partial _t \mp \operatorname{ad}^*_ \xi \right)  \sum _{ j=0}^ {k-1} (-1)^j \partial _t ^j\frac{\delta \ell }{\delta  \xi^{ (j) } }=\mathbf{J} \left( \frac{\delta  \ell}{\delta  q} \right) ,
\end{equation}
where $ \mathbf{J} : T^\ast Q\rightarrow \mathfrak{g}^* $ is again the cotangent bundle momentum map.

\begin{remark}\rm
[Recovering Euler-Poincar\'e equations of \S\ref{subsec_param_dep_Lagr}]\label{left_right_remark} 
$\,$\\  
Equations \eqref{q-dep-eqn} for the control recover the $k^{th}$-order Euler-Poincar\'e equations \eqref{EP_with_q_RL}.
Note that a right, respectively left action of $G$ on $Q$ produces the left, respectively  right Euler-Poincar\'e equations in \eqref{EP_clebsch} consistently with the results in \S\ref{subsec_param_dep_Lagr}. In order to obtain the  Euler-Poincar\'e equations \eqref{EP_with_q_RR} one needs to impose condition ${\rm (A)}'$ instead of ${\rm (A)}$ on the dynamics on $q$.
\end{remark}

\section{Higher-order template matching problems}
\label{hoTemplateMatching-sec}

In this section we generalize the methods of \cite{BrFGBRa2010} to higher order because the added smoothness provided by higher-order models makes them attractive for longitudinal data interpolation, in particular in Computational Anatomy \textit{(CA)}.

We first give a brief account of the previous work done on longitudinal data interpolation in \textit{CA}. Then we derive the equations that generalize \cite{BrFGBRa2010}. After making a few remarks concerning the gain in smoothness, we provide a qualitative discussion of two Lagrangians of interest for \textit{CA}. Finally, we close the section by demonstrating the spline approach to template matching for the finite dimensional case of fitting a spline through a sequence of orientations on $SO(3)$.

\subsection{Previous work on longitudinal data interpolation in CA} \textit{CA} is concerned with modeling and quantifying diffeomorphic evolutions of shapes, as presented in 
\cite{MillerARBE2002,MillerIJCV2001}. Usually one aims at finding a geodesic path, on the space of shapes, between given initial and final data. This approach can be adapted for longitudinal data interpolation; that is, interpolation through a sequence of data points. One may interpolate between the given data points in such a way that the path is piecewise-geodesic, \cite{BegKhanISBI08,durlemann}. It was, however, argued in \cite{TrVi2010} that higher order models, i.e., models that provide more smoothness than the piecewise-geodesic one, are better suited as growth models for typical biological evolutions. As an example of such a higher-order model, spline interpolation on the Riemannian manifold of landmarks was studied there.  In the next paragraph, we will consider another class of models of interest for \textit{CA} that are inspired by an optimal control viewpoint. Indeed, the time-dependent vector field can be seen as a control variable acting on the template and the penalization on this control variable will be directly defined on the Lie algebra. Finally, we underline that this class of models is an interesting alternative to the shape splines model presented in \cite{TrVi2010}.

\subsection{Euler-Lagrange equations for higher-order template matching}\label{some-analytical-remarks-sec}

Let $G$ be a Lie group with Lie algebra $\mathfrak{g}$, and let
\begin{equation}\label{left_representation}
G\times V \rightarrow V , \quad (g, I) \mapsto gI
\end{equation}
be a left representation of $G$ on $V$. Let $\|\cdot\|_V$ be a norm on $V$. We consider minimization problems of the following abstract form:\\

\noindent\textit{Given a Lagrangian $\ell: (k-1) \mathfrak{g} \rightarrow \mathbb{R}$, $\sigma, t_1, \ldots, t_l \in \mathbb{R}$, $T_0, I_{t_1}, \ldots , I_{t_l} \in V$, and $\xi^0_0, \ldots, \xi^{k-2}_0  \in \mathfrak{g}$, minimize the functional
\begin{equation}\label{LongitudinalInterpolation_dash}
E[\xi] := \int_0^{t_l} \ell(\xi(t), \ldots, \xi^{(k-1)}(t))  dt + \frac{1}{2\sigma^2}\sum_{i = 1}^l \left\|g^\xi(t_i)T_0 - I_{t_i}\right\|^2_V
\end{equation}
subject to the conditions $\xi^{(j)}(0) = \xi^j_0$, $j = 0, \ldots, k-2$, where $g^\xi(t_i)$ is the flow of $\xi(t)$ evaluated at time $t_i$. The minimization is carried out over the space $$\mathcal{P}_{k-1} := \left\{ \xi \in C^{k-2}([0,t_l],\mathfrak{g}) \mid \xi^{(k-1)} \in C^{\infty}_{pcw}\left( ([t_i,t_{i+1}])_{i=0}^{l-1} \right) \right\}$$ where $C^{\infty}_{pcw}\left( ([t_i,t_{i+1}])_{i=0}^{l-1} \right)$ denotes the set of piecewise smooth curves whose only discontinuities would be at the $t_i$, $i = 1, \ldots, l-1$, i.e. $$C^{\infty}_{pcw}\left( ([t_i,t_{i+1}])_{i=0}^{l-1} \right) := \left\{ f \in L^2([0,T],\mathfrak{g}) \mid \forall i=0, \ldots, l-1 \; \exists \, f^i \in C^{\infty}([t_i,t_{i+1}],\mathfrak{g}) \text{ s.t. } f = f^i_{|(t_i,t_{i+1})} \right\}.$$
}

\color{black}

More precisely, given such a curve $\xi(t)$ in the Lie algebra $\mathfrak{g}$, its \textit{flow} $g^\xi: t \mapsto g^{\xi}(t) \in G$ is a continuous curve defined by the conditions
\begin{equation}\label{def_flow}
g^{\xi}(0) = e, \quad \text{and} \quad \frac{d}{dt} g^{\xi}(t) = \xi(t) g^{\xi}(t)\,,
\end{equation}
whenever $t$ is in one of the open intervals $(0, t_1), \ldots, (t_{l-1}, t_l)$. Here we used the notation $\xi(t) g^{\xi}(t) :=  TR_{g^\xi(t)} \xi(t)$.
 
We typically think of $(I_{t_1}, \ldots , I_{t_l})$ as the time-sequence of data, indexed by time points $t_j$, $j = 1, \ldots l$, and $T_0$ is the template (the source image). 
Moreover, $\xi: t \in [0, t_l] \mapsto \xi(t) \in \mathfrak{g}$ is typically a time-dependent vector field (sufficiently smooth in time) that generates a flow of diffeomorphisms $g^\xi: t \in [0, t_l] \mapsto g^\xi(t) \in G$. Note that, in this case, the Lie group $G$ is infinite dimensional and a rigorous framework to work in is the \textit{large deformations by diffeomorphisms} setting thoroughly explained in \cite{ty05}.
We will informally refer to this case as the \emph{diffeomorphisms case} or \emph{infinite dimensional case}. 
The expression $g^\xi(t_i)T_0$ represents the template at time $t_i$, as it is being deformed by the flow of diffeomorphisms. Inspired by the second-order model presented in \cite{TrVi2010}, this subsection thus generalizes the work of \cite{BrFGBRa2010} in two directions. First, we allow for a \emph{higher-order} penalization on the time-dependent vector field given by the first term of the functional \eqref{LongitudinalInterpolation_dash}; second, the similarity measure (second term in \eqref{LongitudinalInterpolation_dash}) takes into account \emph{several} time points in order to compare the deformed template with the time-sequence target.

Staying at a general level, we will take the geometric viewpoint of \cite{BrFGBRa2010} in order to derive the Euler-Lagrange equations satisfied by any minimizer of $E$. We suppose that the norm on $V$ is induced by an inner product $\left\langle\,,\right\rangle_V$ and denote by $\flat$ the isomorphism
\[
\flat: V \rightarrow V^*, \quad \omega \mapsto \omega^\flat
\]
that satisfies
\[
\left\langle I, J\right\rangle_V = \left\langle I^\flat, J\right \rangle \quad \mbox{for all } I, J \in V,
\]
where we wrote $\left\langle\,,\right\rangle$ for the duality pairing between  $V$ and its dual $V^*$. The action \eqref{left_representation} of $G$ on $V$ induces an action on $V^*$,
\[
G \times V^*\rightarrow V^*, \quad (g, \omega) \mapsto g\omega
\]
that is defined by the identity
\begin{equation}\label{DualAction_dash}
\left\langle g \omega, I\right\rangle = \left\langle\omega, g^{-1} I\right\rangle \quad \mbox{for all } I \in V,\; \omega \in V^*,\; g\in G.
\end{equation}
The cotangent-lift momentum map $\diamond: V \times V^* \rightarrow \mathfrak{g}^*$ for the action of $G$ on $V$ is defined by the identity
\begin{equation}\label{MomentumMap_dash}
  \left\langle I\diamond \omega, \xi \right\rangle = \left\langle \omega, \xi I\right\rangle, \quad \text{for all} \quad I \in V, \; \omega \in V^*,  \;\xi \in \mathfrak{g},
\end{equation}
where the brackets on both sides represent the duality pairings of the respective spaces for $\mathfrak{g}$ and $V$, and where $\xi I$ denotes the infinitesimal action of $\mathfrak{g}$ on $V$, defined as $\xi\, I :=\left.\frac{d}{dt}\right|_{t=0}g(t)I \in V$ for any $C^1$ curve $g :[-\varepsilon,\varepsilon] \to G$ that satisfies $g(0)=e$ and 
$\left.\frac{d}{dt}\right|_{t=0} g(t) = \xi\in\mathfrak{g}$.
 Note that equations \eqref{DualAction_dash} and 
 \eqref{MomentumMap_dash} imply
\begin{equation}\label{ActionAndMomentumMap_dash}
\operatorname{Ad}^*_{g^{-1}}(I \diamond \omega) = gI \diamond g\omega.
\end{equation}

For the flow defined in \eqref{def_flow}, we also introduce the notation
\begin{equation}\label{Flow_notation}
g^\xi_{t, s} := g^\xi(t)\left(g^\xi(s)\right)^{-1}.
\end{equation}

Lemma 2.4 in \cite{BrFGBRa2010}, which is an adaptation from \cite{via09} and \cite{bmty05}, gives the derivative of the flow at a given time with respect to a variation $(\varepsilon, t) \mapsto \xi(t) + \varepsilon \delta \xi (t) \in \mathfrak{g}$ of a smooth curve $\xi = \xi_0$. Namely,
\begin{equation}\label{variation_gts}
\delta g_{t, s}^{\xi} := \left.\frac{d}{d\varepsilon}\right|_{\varepsilon = 0} g_{t, s}^{\xi_\varepsilon} = g^\xi_{t, s}  \int_s^t \left(\operatorname{Ad}_{g^\xi_{s, r}} \delta \xi(r)\right) dr  \in T_{g^\xi_{t, s}}G.
\end{equation}
Importantly, formula \eqref{variation_gts} also holds for the diffeomorphisms case in a non-smooth setting, as shown in \cite{ty05}, where the assumption is $\xi \in L^2([0,t_l],V)$. Moreover, this proof can be adapted to the case of a finite dimensional Lie group. In particular, formula \eqref{variation_gts} can be used for $\xi \in \mathcal{P}_{k-1}$, whether one works with finite dimensional Lie groups or diffeomorphism groups.

Formula \eqref{variation_gts} and equation 
\eqref{ActionAndMomentumMap_dash} are the key ingredients needed in order to take variations of the similarity measure in \eqref{LongitudinalInterpolation_dash}. With these preparations it is now straightforward to adapt the calculations done in the proof of Theorem 2.5 of \cite{BrFGBRa2010} to our case, in order to show that the following theorem holds.
\begin{theorem} 
\label{ComputationalAnatomyHigherOrderTheorem}
  A curve $\xi\in \mathcal{P}_{k-1}$ is an extremal for the functional $E$, i.e., $\delta E = 0$ if and only if ${\rm (I)}$, ${\rm (II)}$, and ${\rm (III)}$ below hold:
\begin{enumerate}[\rm (I)]
\item For $t$ in any of the open intervals $(0, t_1), \ldots, (t_{l-1}, t_l)$,
\begin{equation}
\sum _{ j=0}^ {k-1} (-1)^j \frac{d^j}{dt^j}  \frac{\delta \ell }{\delta  \xi^{ (j) } }
      = - \sum_{i=1}^l \chi_{[0, t_i]}(t) \left(g^\xi_{t, 0}  T_0 \diamond g^\xi_{t,t_i}  \pi^{i} \right),
\label{mommap-def1}
\end{equation}
where $\pi^i$ is defined by
\[
\pi^i : = \frac{1}{\sigma^2}\left(g^\xi_{t_i, 0} T_0 - I_{t_i}\right)^\flat \in V ^\ast,
\]
and $\chi_{[0, t_i]}$ is the characteristic function of the interval $[0, t_i]$.
  \item For $i = 1, \ldots, l-1$ and $r = 0, \ldots, k-2$,
    \begin{equation}\label{TM_theorem_2}
      \lim_{t\rightarrow t_i^-}\sum_{j \geq r + 1}^{k-1} (-1)^{j - r -1} \frac{d^{j-r -1}}{dt^{j-r -1}} \frac{\delta \ell}{\delta \xi^{(j)}} (t) = \lim_{t\rightarrow t_i^+}  \sum_{j \geq r + 1}^{k-1} (-1)^{j - r -1} \frac{d^{j-r -1}}{dt^{j-r -1}} \frac{\delta \ell}{\delta \xi^{(j)}} (t)\,.
    \end{equation}
  \item For $r = 0, \ldots, k-2$,
    \begin{equation} \label{TM_theorem_3}
       \sum_{j \geq r + 1}^{k-1} (-1)^{j - r -1} \frac{d^{j-r -1}}{dt^{j-r -1}} \frac{\delta \ell}{\delta \xi^{(j)}} (t_l) = 0\,.
    \end{equation}
  \end{enumerate}
\end{theorem}

Note that there is no condition at $t_0=0$ analogous to (III) because of the fixed end point conditions $\xi^{(j)}(0) = \xi^j_0$, for $j = 0, \ldots, k-2$.

\begin{proof}
Set $t_0=0$ for convenience.
A series of partial integrations taking into account the fixed end point conditions $\xi^{(j)}(0) = \xi^j_0$, $j = 0, \ldots, k-2$, leads to
\begin{align}
\delta \int_0^{t_l}\ell  dt 
&= \sum_{i=0}^{l-1} \int_{t_i}^{t_{i+1}} 
\sum_{j=0}^{k-1}\left<\frac{\delta \ell}{\delta \xi^{(j)}}, \delta \xi^{(j)}\right>\, dt \nonumber  \\
&= \sum_{i=0}^{l-1} \int_{t_i}^{t_{i+1}} \left<\sum_{j=0}^{k-1}(-1)^j\frac{d^j}{dt^j} \frac{\delta \ell}{\delta \xi^{(j)}}(t) , \delta \xi(t)\right>\, dt \nonumber \\
&\quad + \sum_{i = 1}^{l-1} \sum_{r=0}^{k-2} \left(
\left< \sum_{j\geq r+1}^{k-1} (-1)^{j-r-1}\left(\frac{d^{j-r-1}}{dt^{j-r-1}} \frac{\delta \ell}{\delta \xi^{(j)}}(t_i^-) -\frac{d^{j-r-1}}{dt^{j-r-1}} \frac{\delta \ell}{\delta \xi^{(j)}}(t_i^+)\right) , \delta \xi^{(r)}(t_i)\right> 
\right. \nonumber \\
&\qquad \qquad \qquad \left. + \left< \sum_{j\geq r+1}^{k-1} (-1)^{j-r-1}\frac{d^{j-r-1}}{dt^{j-r-1}} \frac{\delta \ell}{\delta \xi^{(j)}}(t_l), \delta \xi^{(r)}(t_l)\right> \right)\,.
\label{left-right}
\end{align}
Note that the hypothesis $\xi \in \mathcal{P}_{k-1}$ is sufficient to give meaning to the previous formula.

On the other hand, using formula \eqref{variation_gts} and mimicking the computations done in \cite{BrFGBRa2010}, one finds for the variation of the similarity measure
\begin{equation}
  \delta \left(\frac{1}{2\sigma^2}\sum_{i = 1}^l \left\|g^\xi(t_i)T_0 - I_{t_i}\right\|^2_V\right) = \int_0^{t_l}\left<\sum_{i=1}^l \chi_{[0, t_i]}(t) \left(g^\xi_{t, 0}  T_0 \diamond g^\xi_{t,t_i}  \pi^{i} \right), \delta \xi(t)\right>\, dt\,.
\end{equation}
Assembling the two contributions to $\delta E$, we arrive at
\begin{align}
\delta E 
&= \sum_{s=0}^{l-1} \int_{t_s}^{t_{s+1}} 
\left<\sum_{j=0}^{k-1} (-1)^j\frac{d^j}{dt^j} 
\frac{\delta \ell}{\delta \xi^{(j)}}(t) 
+ \sum_{i=1}^l \chi_{[0, t_i]}(t) \left(g^\xi_{t, 0}  T_0 \diamond g^\xi_{t,t_i}  \pi^{i} \right) , \delta \xi(t)\right>\, dt \nonumber \\
&\quad + \sum_{i = 1}^{l-1} \sum_{r=0}^{k-2}\left(
\left< \sum_{j\geq r+1}^{k-1} (-1)^{j-r-1}
\left(\frac{d^{j-r-1}}{dt^{j-r-1}} 
\frac{\delta \ell}{\delta \xi^{(j)}}(t_i^-) 
-\frac{d^{j-r-1}}{dt^{j-r-1}} 
\frac{\delta \ell}{\delta \xi^{(j)}}(t_i^+)\right) , 
\delta \xi^{(r)}(t_i)\right> \right. \nonumber \\
&\qquad \qquad\qquad \left. 
+ \sum_{r=0}^{k-2} \left< \sum_{j\geq r+1}^{k-1} (-1)^{j-r-1}\frac{d^{j-r-1}}{dt^{j-r-1}} \frac{\delta \ell}{\delta \xi^{(j)}}(t_l), \delta \xi^{(r)}(t_l)\right>\right)\,.
\end{align}
Stationarity $\delta E = 0$ therefore leads to equations \eqref{mommap-def1} - \eqref{TM_theorem_3}.
\end{proof}
\begin{remark}\label{HOEP-remark}{\rm 
The right-hand side of equation \eqref{mommap-def1} follows coadjoint motion on every open interval $(0, t_1), \ldots, (t_l, t_{l-1})$. Therefore,
\begin{equation}\label{TM_HO_EP}
  \left(\frac{d}{dt} + \operatorname{ad}^*_{\xi(t)}\right) \sum _{ j=0}^ {k-1} (-1)^j \frac{d^j}{dt^j}  \frac{\delta \ell }{\delta  \xi^{ (j) } }= 0,
\end{equation}
in which we once again recognize the higher-order Euler-Poincar\'{e} equation \eqref{EP_k_dash}.}
\end{remark}

\subsection{Two examples of interest for computational anatomy}\label{two_examples_CA}
Regarding potential applications in \textit{CA}, an interesting property of higher-order models is the gain in smoothness of the optimal path $T: t \in [0, t_l] \mapsto g^\xi(t)T_0 \in V$, in comparison with first-order models. For instance, in the case of piecewise-geodesic (i.e., first-order) interpolation, where $\ell(\xi) := \frac{1}{2}\|\xi\|_{\mathfrak{g}}^2$,
equation  \eqref{mommap-def1} reads
\begin{equation}
  \xi(t) =  - \sum_{i=1}^l \chi_{[0, t_i]}(t) \left(g^\xi_{t, 0} T_0 \diamond g^\xi_{t,t_i}  \pi^{i} \right) .
\end{equation}
In general therefore, $\xi$ will be discontinuous at each time point $t_i$ for $i < l$, which implies non-differentiability of $T$ at these points. In contrast, for the Lagrangian $\ell_1(\dot{\xi}) := \frac{1}{2} \|\dot{\xi}\|_{\mathfrak{g}}^2$, equation  \eqref{mommap-def1} becomes
\begin{equation}
  \ddot{\xi}(t) = \sum_{i=1}^l \chi_{[0, t_i]}(t) \left(g^\xi_{t, 0} T_0 \diamond g^\xi_{t,t_i}  \pi^{i} \right) .
\end{equation}
Now the curves $\xi(t)$ and $T(t)$ are $C^1$ and $C^2$ on $[0, t_l]$, respectively.
Note that the inexact interpolation we consider here yields a $C^2$ curve $T$, whereas the exact interpolation method presented in the example of Section \ref{sec 2.2} leads to $C^1$ solutions.
Note also that the minimization of the functional $E$ for $\ell_1$ when $l=1$ produces Lie-exponential solutions on $G$.
More precisely, if the Lie-exponential map is surjective and the action of $G$ on $V$ is transitive, then there exists $\xi_0 \in \mathfrak{g}$ such that $\exp(t_1 \xi_0) T_0 = I_{t_1}$. Hence, the constant curve $\xi \equiv \xi_0$ is a minimizer of the functional $E$, with $E[\xi] = 0$. 
The Lie-exponential has been widely used in \textit{CA}, for instance in \cite{ArsignyMRM06,AshburnerNI2007}.

Another Lagrangian of interest for \textit{CA} is $\ell_2(\xi,\dot{\xi}) := \frac{1}{2}\|\dot{\xi} + \operatorname{ad}^\dagger_{\xi} \xi\|_{\mathfrak{g}}^2$, which measures the acceleration on the Lie group for the right-invariant metric induced by the norm $\|\cdot \|_{\mathfrak{g}}$.
The Lagrangian $\ell_2$ may therefore have more geometrical meaning than $\ell_1$. However, $\ell_1$ is worth studying since it is simpler from both {\color{black} the} computational and {\color{black} the} analytical point of view: The existence of a minimizer for $\ell_1$ {\color{black} can} be obtained straightforwardly following the strategy of \cite{ty05}. In contrast, a deeper analytical study is required for $\ell_2$, since analytical issues arise in infinite dimensions.

\subsection{Template matching on the sphere}

Consider as a finite-dimensional example $G = SO(3)$ with norm $\|\mathbf{\Omega}\|_{\mathfrak{so}(3)} = 
\sqrt{\mathbf{\Omega} \cdot \mathbf{I} \mathbf{\Omega}}$ on the Lie algebra $\mathfrak{so}(3)$, where $\mathbf{I}$ is a symmetric positive-definite matrix (the moment of inertia tensor). Let $V = \mathbb{R}^3$ with $\|\cdot\|_{\mathbb{R}^3}$ the Euclidean distance. We would like to interpolate a time sequence of points on the unit sphere $S^2 \subset \mathbb{R}^3$ starting from the template $T_0 = \left(
  \begin{array}{c}
    1\\
    0\\
    0\\
  \end{array}\right)$
. Choose the times to be $t_i = \frac{1}{5} i$ for $i = 1, \ldots ,5$, and
\begin{equation}
I_{t_1} = \left(
  \begin{array}{c}
    0\\
    1\\
    0\\
  \end{array}\right), \quad I_{t_2} = \left(
  \begin{array}{c}
    0\\
    0\\
    1\\
  \end{array}\right), \quad I_{t_3} =\frac{1}{\sqrt{2}}\left(
  \begin{array}{c}
    1\\
    0\\
    1\\
  \end{array}\right), \quad I_{t_4} =  \frac{1}{\sqrt{2}}\left(
  \begin{array}{c}
    1\\
    1\\
    0\\
  \end{array}\right), \quad I_{t_5} = \frac{1}{\sqrt{3}} \left(
  \begin{array}{c}
    1\\
    1\\
    1\\
  \end{array}\right).
\end{equation}

The associated minimization problem for a given Lagrangian $\ell(\mathbf{\Omega}, \ldots, \mathbf{\Omega}^{(k-1)})$ is:\\
Minimize
\begin{equation}\label{splines_SO3}
E[\mathbf{\Omega}] := \int_0^1 \ell(\mathbf{\Omega}(t), \ldots, \mathbf{\Omega}^{(k-1)}(t))  dt + \frac{1}{2\sigma^2}\sum_{i = 1}^5 \left\|\Lambda^\mathbf{\Omega}(t_i)T_0 - I_{t_i}\right\|^2_{\mathbb{R}^3},
\end{equation}
subject to the conditions $\mathbf{\Omega}^{(j)}(0)= \mathbf{\Omega}^j_0$, $j= 0,\ldots, k-2$, where $\Lambda^{\mathbf{\Omega}}(t)$ is a continuous curve defined by
\[\Lambda^{\mathbf{\Omega}}(0) = e\,, \quad \text{and} \quad \frac{d}{dt}\Lambda^{\mathbf{\Omega}}(t)= \Omega(t)\Lambda^{\mathbf{\Omega}}(t)\,,
\]
whenever $t$ is in one of the open intervals $(0, t_1), \ldots, (t_4, t_5)$. As we mentioned in Section \ref{two_examples_CA}, an important property of higher-order models is the increase in smoothness of the optimal path when compared with first-order models. We illustrate this behavior in Figures \ref{FirstOrderTemplateMatching} and \ref{HOTemplateMatching}:

Figure \ref{FirstOrderTemplateMatching} shows the interpolation between the given points $I_{t_1}, \ldots, I_{t_5}$ for the first order Lagrangian
\begin{equation}\label{First_order_so3_Lagrangian}
\ell(\mathbf{\Omega}) = \frac{1}{2} \mathbf{\Omega} \cdot \mathbf{I} \mathbf{\Omega}.
\end{equation}
We contrast this with the second order model
\begin{equation}\label{Second_order_so3_Lagrangian}
\ell(\mathbf{\Omega}, \dot{\mathbf{\Omega}}) = \frac{1}{2} \left(\dot{\mathbf{\Omega}} + \mathbf{I}^{-1}(\mathbf{\Omega} \times \mathbf{I}\mathbf{\Omega})\right) \cdot \mathbf{I}  \left(\dot{\mathbf{\Omega}} + \mathbf{I}^{-1}(\mathbf{\Omega} \times \mathbf{I}\mathbf{\Omega})\right)\,.
\end{equation}
Note that this is the reduced Lagrangian for splines on $SO(3)$, as we discussed in Section \ref{hoEP-sec}, and for $\mathbf{I}=e$  we recognize equation \eqref{TM_HO_EP} to be the NHP equation \eqref{NHP_equation}.

Figure \ref{HOTemplateMatching} visualizes the resulting interpolation for two different choices of the moment of inertia tensor $\mathbf{I}$, namely
\begin{equation}
  \mathbf{I}_1 :=\left(
      \begin{array}{ccc}
        1 & 0 & 0\\
        0&1&0\\
        0&0&1\\
     \end{array}\right) \quad \mbox{and} \quad \mathbf{I}_2 :=\frac{1}{\sqrt{2}}\left(
      \begin{array}{ccc}
        1 & 0 & 0\\
        0&2&0\\
        0&0&1\\
     \end{array}\right).
\end{equation}
In order to compare the two cases we have normalized $\mathbf{I}_2$ in such a way that it has the same norm as $\mathbf{I}_1$ with respect to the norm $\|\mathbf{I}\|^2 =   \operatorname{tr}(\mathbf{I}^{\mathsf{T}}\mathbf{I})$. The figures were obtained by minimizing the functional $E$ using the downhill simplex algorithm $fmin\_tnc$ that is included in the $optimize$ package of SciPy, \cite{SciPy}.

\begin{figure}[h!t]
\begin{center}
\subfigure[$\sigma = 0.18$]{\begin{overpic}[scale=0.4]%
{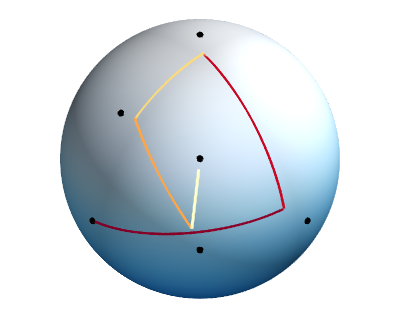}
\put(13, 20){$T_0$}
\end{overpic}
} \hspace{0pt}
\subfigure[$\sigma = 0.01$]{\begin{overpic}[scale=0.4]%
{MI1-1-1sigma0dot01firstorder.png}
\put(13, 20){$T_0$}
\end{overpic}}
\end{center}
\caption{\footnotesize{First order template matching results are shown for the Lagrangian \eqref{First_order_so3_Lagrangian} with $\mathbf{I} = e$, for two different values of tolerance $\sigma$. These values have been chosen so that the sum of the mismatch penalties is similar in size to the one obtained in the second order template matching shown in Figure \ref{HOTemplateMatching}. As might be expected,  when the tolerance is smaller, the first order curves pass nearer their intended target points. These first order curves possess jumps in tangent directions at the beginning of each new time interval. }}
\label{FirstOrderTemplateMatching}
\end{figure}


\begin{figure}[h!t]
\begin{center}
  \subfigure[$\sigma = 0.05$, $\mathbf{I} = \mathbf{I}_1$]{\begin{overpic}[scale=0.4]%
{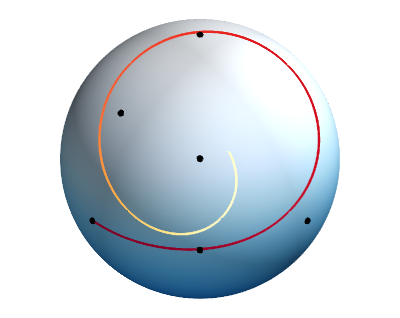}
\put(13, 20){$T_0$}
\end{overpic}
} \hspace{0pt}
\subfigure[$\sigma = 0.001$, $\mathbf{I} = \mathbf{I}_1$]{\begin{overpic}[scale=0.4]%
{MI1-1-1sigma0dot001.png}
\put(13, 20){$T_0$}
\end{overpic}} \\
  \subfigure[$\sigma = 0.05$, $\mathbf{I} = \mathbf{I}_2$]{\begin{overpic}[scale=0.4]%
{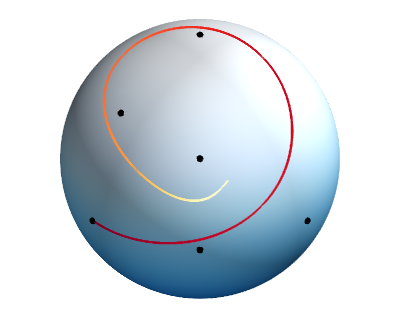}
\put(13, 20){$T_0$}
\end{overpic}
} \hspace{0pt}
\subfigure[$\sigma = 0.001$, $\mathbf{I} = \mathbf{I}_2$]{\begin{overpic}[scale=0.4]%
{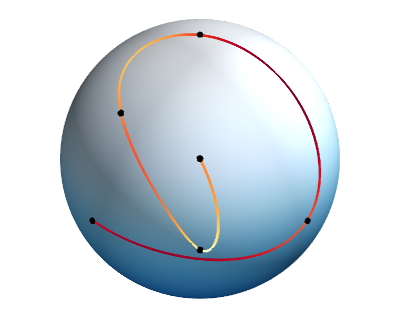}
\put(13, 20){$T_0$}
\end{overpic}}
\end{center}
\caption{\footnotesize{{The pictures in the top row show the template matching for the Lagrangian \eqref{First_order_so3_Lagrangian} with  $\mathbf{I}_1$ with two different values of tolerance, $\sigma$. The bottom row represents the corresponding matching results for $\mathbf{I}_2$. One observes that the quality of  matching increases as the tolerance decreases. This is due to the increased weight on the penalty term in \eqref{LongitudinalInterpolation_dash}. The color of the curves represents the magnitude of the velocity vector of the curve on the sphere (red is large, white is small). We fixed the initial angular velocity $\mathbf{\Omega}(0)=\frac{5 \pi}{2}\left(0, 0, 1\right)$.  On comparing these figures with those in the first order case, one observes that the second-order method produces smoother curves.}}}
\label{HOTemplateMatching}
\end{figure}

\clearpage

\begin{remark}{\rm 
  Standard variational calculus arguments ensure the existence of a minimizer to to the functional \eqref{LongitudinalInterpolation_dash} with Lagrangian \eqref{Second_order_so3_Lagrangian}. In Theorem \ref{ComputationalAnatomyHigherOrderTheorem}, we chose to fix $\xi^{(j)}(0) = \xi^j_0$, $j = 0, \ldots, k-2$, which reduces in this case to fixing $\mathbf{\Omega}(0)$. We might, however, also want to optimize over this initial velocity. Unfortunately, examples can be exhibited where there does not exist any solution to the minimization of $E$ if one also minimizes over $\mathbf{\Omega}(0)$. One possibility to restore well-posedness while retaining the minimization over $\mathbf{\Omega}(0)$ is to modify $E$ by adding a penalty on the norm $\|\mathbf{\Omega}(0)\|$.
  
The situation in infinite dimensions is similar, however proving existence results would require much deeper analytical study than in the finite dimensional case.}
\end{remark}
In this section we presented higher-order methods that increase the smoothness in interpolating through a sequence of data points. In future work these methods will be compared to the shape spline model introduced in \cite{TrVi2010}.
Also of interest for \textit{CA} is the metamorphosis approach that is discussed briefly in Section \ref{optimizDyn-sec}.

\section{Optimization with penalty} \label{optimizDyn-sec}

This section adapts the optimization approach of \cite{GBHoRa2010} to higher-order Lagrangians. 
As in the case of the Clebsch-Pontryagin approach, one considers the (right or left) action $\Phi:G\times Q\rightarrow Q$ of a Lie group $G$ on a manifold $Q$. The basic idea is to replace the constraints in the Clebsch optimal control problem with a penalty function added to the cost function and to obtain in this way a classical (unconstrained) optimization problem. The penalty term is expressed with the help of a Riemannian metric $\gamma$ on the manifold $Q$. Given a cost function $\ell: k \mathfrak{g} \times Q \rightarrow \mathbb{R} $, $\sigma>0$ and the elements $n_0, n_T\in Q$, $\xi^j_0, \xi^j_T\in\mathfrak{g}$, $j=0,...,k-2$,  one minimizes
\begin{equation}\label{optimization_penalty}
\int _{ 0}^{T}\left(\ell\left(\xi,\dot{\xi},\ldots, 
\xi^{(k-1)}, n\right)
+\frac{1}{2\sigma^2}\|\dot n-\xi_Q(n)\|^2 \right) dt,
\end{equation}
over curves $t\mapsto n(t)\in Q$ and $ t\mapsto \xi(t) \in \mathfrak{g}$ such that
\[
n(0)=n_0,\quad n(T)=n_T,\quad \xi^{(j)}(0)=\xi^j_0,\quad \xi^{(j)}(T)=\xi^j_T,\;\; j=0,...,k-2,
\]
where $\|\cdot\|$ is the norm on $TQ$ induced by the metric  $\gamma$ and, as in the Clebsch-Pontryagin case in Section \ref{Clebsch-sec}, $\xi_Q(n)$ denotes the infinitesimal generator of the $G$-action associated to 
$\xi\in \mathfrak{g}$, evaluated at $n\in Q$.
The corresponding stationarity conditions are found to be:
\begin{equation}\label{stationarity_cond_pi}
\sum_{j=0}^{k-1}(-1)^j\partial _t ^j \frac{\delta\ell}{\delta\xi^{ (j) }}=\mathbf{J}(\pi),\quad \dot n=\xi_Q(n)+\sigma^2\pi^\sharp,\quad \frac{D}{Dt}\pi=- \langle\pi,\nabla\xi_Q\rangle+\frac{\partial \ell}{\partial n},
\end{equation}
where the notation
\[
\pi:=\frac{1}{\sigma^2}\nu_n^\flat=\frac{1}{\sigma^2}\left(\dot n-\xi_Q(n)\right)^\flat\in T^*Q
\]
has been used and the covariant derivatives $D/Dt$ and $ \nabla $ are associated to the Riemannian metric $ \gamma $ on $Q$.

These equations should be compared with the stationarity conditions \eqref{stationarity_conditions_Clebsch} associated to the Clebsch approach, which can be rewritten, with the help of a Riemannian metric, as
\begin{equation}\label{stationarity_cond_alpha}
\sum_{j=0}^{k-1}(-1)^j\partial _t ^j \frac{\delta\ell}{\delta\xi^{ (j) }}=\mathbf{J}(\alpha),\quad \dot q=\xi_Q(q),\quad \frac{D}{Dt}\alpha=- \langle\alpha,\nabla\xi_Q\rangle+\frac{\partial \ell}{\partial q}.
\end{equation}

Before proceeding further, we will pause to define some additional notation that will be convenient later. 

\begin{definition}\label{def_F_nabla} Consider a Lie group $G$ acting on a Riemannian manifold $(Q, \gamma)$. We define the $\mathfrak{g}^*$-valued $(1,1)$ tensor field $\mathcal{F}^\nabla:T^*Q\times TQ\rightarrow\mathfrak{g}^*$ associated to the Levi-Civita connection $\nabla$ by
\begin{equation}\label{definition_F_nabla}
\left\langle\mathcal{F}^\nabla(\alpha_q,u_q),\eta\right\rangle:=\left\langle\alpha_q,\nabla_{u_q}\eta_Q(q)\right\rangle,
\end{equation}
for all $u_q\in T_qQ$, $\alpha_q\in T^*_qQ$, and $\eta\in\mathfrak{g}$.
\end{definition}

The main properties of the tensor field $\mathcal{F}^\nabla$ are discussed in \cite{GBHoRa2010}, where one also finds the proofs of the following two lemmas about the properties of $\mathcal{F}^\nabla$. The first lemma below relates $\mathcal{F}^\nabla$ to the connectors of the covariant derivatives on $TQ$ and $T^*Q$. The second lemma explains that $\mathcal{F}^\nabla$ is antisymmetric under transposition in the inner product defined by the Riemannian metric $\gamma$ when $G$ acts by isometries.

\begin{lemma}\label{property_F_connector} 
For all $\alpha_q\in T^*_qQ$, $u_q\in T_qQ$, and $\xi\in\mathfrak{g}$, 
\[
\left\langle \mathcal{F}^\nabla(\alpha_q,u_q), \xi \right\rangle =\left\langle\alpha_q,K(\xi_{TQ}(u_q))\right\rangle=-\left\langle K(\xi_{T^*Q}(\alpha_q)),u_q\right\rangle,
\]
where $K$ denotes the connectors of the covariant derivatives on $TQ$ and $T^*Q$, respectively.
\end{lemma}

\begin{proof} See the proof of Lemma 3.5 in \cite[\S3]{GBHoRa2010}.
\end{proof}

\begin{remark}{\rm 
A detailed treatment of connectors and their associated linear connections for covariant derivatives can be found in \cite[\S13.8]{Michor2008}. We also refer to \cite{GBHoRa2010} for useful properties of the connector $K$ of relevance to the present paper. 
In infinite dimensions one needs to assume that the given weak Riemannian metric has a smooth geodesic spray  $S \in \mathfrak{X}(TG)$, but such analytical issues will not be of concern to us here. }
\end{remark}

\begin{lemma}\label{antisymmetry} If $G$ acts by isometries, then $\mathcal{F}^\nabla$ is antisymmetric, that is
\[
\mathcal{F}^\nabla(\alpha_q,u_q)=-\mathcal{F}^\nabla(u_q^\flat,\alpha_q^\sharp),
\]
for all $u_q\in T_qQ$, $\alpha_q\in T^*_qQ$.
\end{lemma}
\begin{proof}
Since $G $ acts by isometries, $\boldsymbol{\pounds}_{ \xi_Q}g = 0$; which implies  $(\nabla\xi_Q)^T=-\nabla\xi_Q$.
\end{proof}

The tensor field $\mathcal{F}^\nabla$ arises naturally in computing the equations of motion associated to the stationarity conditions (\ref{stationarity_cond_pi}) for optimization with penalty. A computation, similar to the one given in \cite[\S3]{GBHoRa2010} in the first order case, yields 
\begin{equation}\label{unconstrained_A}
\left\{
\begin{array}{l}
\vspace{0.2cm}\displaystyle 
\left(\partial _t \mp \operatorname{ad}^*_ \xi \right) \sum_{j=0}^{k-1}(-1)^j\partial _t ^j \frac{\delta\ell}{\delta\xi^{ (j) }}= \mathbf{J}\left(\frac{\partial \ell}{\partial n}\right)
+\frac{1}{\sigma^2}\mathcal{F}^\nabla(\nu_n^\flat,\nu_n),\\
\displaystyle \frac{D}{Dt}\nu_n^\flat=-\langle\nu_n^\flat,\nabla\xi_Q\rangle+\sigma^2\frac{\partial \ell}{\partial n}, \quad \nu_n : = \dot n-\xi_Q(n)
,
\end{array}\right.
\end{equation}
where in $(\mp)$ one chooses $-$ (resp. $+$) when $G$ acts on $Q$ by a right (resp. left) action, consistently with \eqref{q-dep-eqn}.

As a consequence of Lemma \ref{antisymmetry}, if $G$ acts on $(Q, \gamma )$ by isometries, then the term $\frac{1}{\sigma^2}\mathcal{F}^\nabla(\nu_n^\flat,\nu_n)$ vanishes so that the optimization problem \eqref{optimization_penalty} produces the $k^{th}$ order Euler-Poincar\'e equations.

\paragraph{Example: The NHP equation via optimization.} In this case, since $SO(3)$ acts by isometries on $ \mathbb{R}  ^3 $, the minimization problem
\[
\min\int_{t _0 }^{ t _1 } \left( \frac{1}{2} \| \dot {\mathbf{\Omega}} \|^2+ \frac{1}{2 \sigma ^2 } \| \dot {\mathbf{q}}- \mathbf{\Omega} \times \mathbf{q} \| ^2 \right) dt
\]
produces the NHP equations \eqref{Noakes-eqn}.

\paragraph{Metamorphosis and Lagrange-Poincar\'e reductions.}
Equations \eqref{unconstrained_A} may also be obtained by a generalization of the metamorphosis reduction developed  in \cite{GBHoRa2010}, as follows. For simplicity, we only treat the case of a right action of $G$ on $Q$.

Consider a $G$-invariant Lagrangian $L=L\left( g, \dot{ g}, ..., g^{ (k) }, q, \dot{ q} \right) :T^{(k)}G\times TQ\rightarrow \mathbb{R}$ relative to the action of $h\in G$ given by
\[
\left( g, \dot{ g}, ..., g^{ (k) }, q, \dot{ q} \right) \mapsto \left( hg, h\dot{ g}, ..., hg^{ (k) }, qh ^{-1} , \dot{ q} h ^{-1}  \right) 
\]
and consider the quotient map
\begin{equation}\label{metamorphosis_quotient}
\left( g, \dot{ g}, ..., g^{ (k) }, q, \dot{ q} \right)\mapsto \left( \xi , \dot{ \xi },...,  \xi ^{ (k-1) },n, \nu \right) \in k \mathfrak{g}  \times TQ,\quad \xi = g ^{-1} \dot g ,\; n= q g , \;\nu = \dot q g.
\end{equation}
The equations of motion for the reduced Lagrangian $\ell_{M}$ induced by $L$ on $k \mathfrak{g}  \times TQ$ can be obtained by a direct generalization of the method used in \cite{GBHoRa2010} for $k=1$. If $L$ has the particular form
\[
L\left( g, \dot{g}, ..., g^{ (k) }, q, \dot{q} \right) =\mathcal{L} \left( g, \dot{g}, ..., g^{ (k) },q\right) + \frac{1}{2\sigma ^2 }\|\dot q g \|^2,
\]
where $ \mathcal{L} $ is the $G$-invariant Lagrangian associated to the function $\ell$ in \eqref{optimization_penalty}, then we recover equations \eqref{unconstrained_A} (with the upper sign chosen).

Instead of the so-called metamorphosis quotient map \eqref{metamorphosis_quotient} one may also use Lagrange-Poincar\'e reduction with the quotient map
\begin{equation}\label{LP_quotient}
\left( g, \dot{ g}, ..., g^{ (k) }, q, \dot{ q} \right)\mapsto \left( \xi , \dot{ \xi },..., \xi^{ (k-1) } ,n,\dot n\right) \in k \mathfrak{g} \times TQ ,\quad \xi = g ^{-1} \dot{g} ,\quad n= q g.
\end{equation}
The reduced equations of motion for metamorphosis with geometric splines that arise in the Lagrange-Poincar\'e approach are
\begin{equation} 
\left\{
\begin{array}{l}
\displaystyle \vspace{0.2cm}   \frac{\delta \ell _{LP}}{\delta  n}-\frac{d}{dt} \frac{\delta \ell _{LP}}{\delta  \dot{n}}=0,\\
\displaystyle\left(\partial _t \mp \operatorname{ad}^*_ \xi \right) \sum_{j=0}^{k-1}(-1)^j\partial _t ^j \frac{\delta\ell_{LP}}{\delta\xi^{ (j) }} =0,
\end{array}
\right.
\end{equation} 
(with the upper sign chosen) where $\ell_{LP}$ is the reduced Lagrangian associated to the same unreduced Lagrangian $L$ as before, but using the quotient map \eqref{LP_quotient} instead of \eqref{metamorphosis_quotient}.

Note that the Lagrange-Poincar\'e approach generalizes easily to higher-order Lagrangians in $q$
such as  $L:=L \left( g, \dot{g}, ..., g^{ (k) }, q, \dot{q}, ..., q^{ (k) } \right) : T^{ (k) }(G\times Q) \rightarrow \mathbb{R}  $. The equations of motions are then simply
\begin{equation}\label{LP_higher_order}
\left\{
\begin{array}{l}
\displaystyle \vspace{0.2cm}   \sum_{j=0}^{k-1}(-1)^j\frac{\delta\ell_{LP}}{\delta n^{ (j) }} =0,\\
\displaystyle\left(\partial _t  \pm \operatorname{ad}^*_ \xi  \right) \sum_{j=0}^{k-1}(-1)^j\partial _t ^j \frac{\delta\ell_{LP}}{\delta\xi^{ (j) }} =0.
\end{array}
\right.
\end{equation}

The metamorphosis reduction approach also generalizes to higher higher-order Lagrangians in $q$. In this case, one uses the quotient map
\begin{equation}\label{metamorphosis_quotient2}
\left( g, \dot{ g}, ..., g^{ (k) }, q, \dot{q}, ..., q^{ (k) } \right) \mapsto \left( \xi , \dot{ \xi },..., \xi^{ (k-1) } ,n, \nu_1,..., \nu _k  \right) \in k \mathfrak{g}  \times T^{ (k) }Q,
\end{equation}
where $\xi = g ^{-1} \dot g$ and  $ \left( n, \nu_1,..., \nu _k \right) = \Phi ^{ (k) }_g  \left( q, \dot{q}, ..., q^{ (k) } \right)$, $ \Phi ^{ (k) }$ being the natural induced action of $G$ on $T^{ (k) }Q$. However for $k \geq 2$, the associated reduced equations are quite complex on general Riemannian manifolds so one may prefer to use the equivalent Lagrange-Poincar\'e formulation \eqref{LP_higher_order}.

\begin{remark}{\rm 
The idea of metamorphosis with splines may apply in imaging as in \cite{HoTrYo2009} by using, e.g., $L( g_t, \dot{ g} _t , \ddot{ g} _t , \eta _t , \dot{\eta }_t)$, $L( g_t, \dot{ g} _t  , \eta _t , \dot{\eta }_t,  , \ddot{ \eta } _t)$, or $L( g_t, \dot{ g} _t ,\ddot{ g} _t, \eta _t , \dot{\eta }_t, \ddot{ \eta } _t)$, instead of $L( g_t, \dot{ g} _t, \eta _t , \dot{\eta }_t)$.}
\end{remark}

\section{Clebsch and Lie-Poisson-Ostrogradsky formulations}\label{Ham-forms-sec}

In this Section we present two Hamiltonian formulations associated to the higher order Euler-Poincar\'e equations \eqref{EP_with_q_RL} with $q$-dependence. (The case of equations \eqref{EP_with_q_RR} may be obtained by making obvious modifications.) The first is a canonical Hamiltonian formulation that generalizes to higher order the canonical Clebsch formulation of Euler-Poincar\'e dynamics. The second is a generalization of the Lie-Poisson formulation (with $q$-dependence) to higher order, that uses Ostrogradsky momenta. We now recall these formulations in the first order case.

\paragraph{Clebsch canonical formulation.} This is associated to the optimal control formulation described in \S\ref{Clebsch-sec}. In the case $k=1$ the canonical Hamiltonian formulation is already given by the Pontryagin approach. Indeed, if $ \xi \mapsto \frac{\delta \ell }{\delta \xi  }$ is a diffeomorphism we consider the function $h: \mathfrak{g}  \times Q \rightarrow \mathbb{R}$ defined by
\[
h(\mu ,q):= \left\langle \mu , \xi \right\rangle - \ell( \xi ,q),\quad \frac{\delta\ell}{\delta \xi}= \mu
\]
and the collective Hamiltonian $H:T^*Q\rightarrow\mathbb{R}$ given by $H( \alpha _q):= h(\mathbf{J}(\alpha_q ),q)$.
If $ \alpha_q(t)$ is a solution of  Hamilton's canonical equations for $H$ on $T^*Q$, then $(\mu(t), q(t))$, where
$\mu(t):=\mathbf{J}(\alpha_q(t))$,
is a solution of the Euler-Poincar\'e equations
\begin{equation}\label{EP_eq_q}
\left( \partial _t  \pm \operatorname{ad}^*_ \xi \right) \frac{\delta\ell}{\delta \xi} = \mathbf{J} \left( \frac{\delta\ell}{\delta q } \right).
\end{equation}

This canonical formulation of the Euler-Poincar\'e equations recovers some important examples such as the Clebsch variables for the ideal fluid \cite{MaWe1983}, singular solutions of the Camassa-Holm equations \cite{HoMa2004}, and double bracket equations, as explained in \cite{GBRa2009}.

\paragraph{Lie-Poisson formulation.} This is obtained by reduction of the Hamiltonian $H_{q _0 }:T^*G \rightarrow \mathbb{R}  $ associated to $L_{ q _0 }$ by Legendre transformation. If $L$ is $G$-invariant as a function defined on $TG \times Q$, then $H:T^*G \times Q \rightarrow \mathbb{R}$ is $G$-invariant and therefore induces the Hamiltonian $h$ given above. By Poisson reduction of the manifold $T^*G \times Q$, where $Q$ is endowed with the trivial Poisson structure, one obtains the Lie-Poisson equations
\[
\left( \partial _t  \pm \operatorname{ad}^*_ { \frac{\delta h }{\delta  \mu} } \right) \mu  = -\mathbf{J} \left( \frac{\delta h}{\delta q } \right),\quad \partial _t  q-\left(\frac{\delta h}{\delta
\mu}\right)_Q(q)=0,
\]
together with the associated Poisson structure
\begin{equation}\label{red_Poisson}
\{f,g\}(\mu,q)=\pm \left\langle\mu,\left[\frac{\delta
f}{\delta\mu},\frac{\delta g}{\delta\mu}\right]\right\rangle+\left\langle
\mathbf{J}\left(\frac{\delta f}{\delta q}\right),\frac{\delta g}{\delta\mu}\right\rangle-\left\langle
\mathbf{J}\left(\frac{\delta g}{\delta q}\right),\frac{\delta f}{\delta\mu}\right\rangle
\end{equation}
on $\mathfrak{g}^*\times Q$; see \cite{GBTr2010}. These equations are equivalent to their Lagrangian counterpart \eqref{EP_eq_q}.

\subsection{Higher order Clebsch formulations}

\subsubsection*{Second order case}

Recall from Defintion \ref{def_Clebsch} that the Clebsch-Pontryagin variational formulation of the second order Euler-Poincar\'e equations reads
\[
\delta \int _{ t_0 }^{t _1 }\left(  \ell( \xi , \dot \xi,q)+\langle\alpha,\dot{q}-\xi  _Q (q)\rangle \right) dt=0,
\]
over curves $\xi (t)\in \mathfrak{g}$ and $\alpha (t)\in T_{q(t)}^*Q$ and under conditions $(A), (B), (C)$. If $ \dot \xi \mapsto \pi  :={\delta  \ell}/{\delta \dot  \xi }$ is a diffeomorphism,  we define $h( \xi , \pi  ,q):= \langle \mu , \dot \xi \rangle - \ell( \xi , \dot \xi ,q)$ and the Pontryagin variational principle may be written equivalently as
\[
\delta \int _{ t_0 }^{t _1 }\left(  \langle \pi  , \dot \xi \rangle - h\left( \xi , \pi  ,q\right)+\langle\alpha,\dot{q}-\xi  _Q (q)\rangle \right) dt=0,\quad \hbox{where}\quad \pi :=\frac{\delta  \ell}{\delta \dot\xi  } 
\]
over curves $\xi (t)\in \mathfrak{g}  $ and $\alpha (t)\in T^*_{q(t)}Q$. Equivalently, this can be reformulated as
\[
\delta \int _{ t_0 }^{t _1 }\left(  \langle \pi  , \dot \xi \rangle - h\left( \xi , \pi ,q \right)+\langle\alpha,\dot{q}-\xi  _Q (q)\rangle \right) dt=0,
\]
over curves $ \xi (t)\in \mathfrak{g}$, $\pi  (t)\in \mathfrak{g}^\ast$, and $\alpha (t)\in T^*_{q(t)}Q$, 
where $\pi  (t)$ is now an independent curve. The relation $\pi ={\delta  \ell}/{\delta \dot\xi  } $ is recovered by variations of  $ \mu (t)$. One observes that this is simply the usual Hamilton Phase Space Variational Principle (i.e., not a Pontryagin Maximum Principle) on the phase space $T^*(Q\times \mathfrak{g}  )$
\[
\delta \int _0^1H(\alpha, \xi , \pi )- \left\langle (\alpha , \pi  ),(\dot q, \dot \xi ) \right\rangle
dt=0, 
\]
for the Hamiltonian
\[
H: T^*(Q\times \mathfrak{g}  ) \rightarrow \mathbb{R}  , \quad H(\alpha_q  ,\xi , \pi  ):=h( \xi , \pi ,q )+ \left\langle \mathbf{J} ( \alpha_q  ), \xi \right\rangle.
\]
We thus have proved the following result.

\begin{theorem}\label{thm_can_form} Let $\ell: 2 \mathfrak{g}\times Q  \rightarrow \mathbb{R} $, $\ell=\ell( \xi , \dot \xi,q )$ be a cost function such that $ \dot \xi \mapsto \pi  :={\delta  \ell}/{\delta \dot  \xi }$ is a diffeomorphism and define the function
\begin{equation}\label{xi_dot_leg_transform}
h( \xi , \pi ,q ):= \langle \pi  , \dot \xi \rangle - \ell( \xi , \dot \xi,q ).
\end{equation}
Then the stationarity conditions \eqref{stationarity_conditions_Clebsch_q} for the $2^{nd}$ order Clebsch-Pontryagin optimal control problem \eqref{Clebsch_optimal_control} with cost function $\ell$ are given by the canonical Hamilton equations on $T^*(Q\times \mathfrak{g} )$ relative to the Hamiltonian $H( \xi , \mu , \alpha )=h( \xi , \pi ,q )+ \left\langle \mathbf{J} ( \alpha ), \xi \right\rangle $.
\end{theorem}

One can alternatively prove this result by computing explicitly the canonical Hamilton equations for $H$ on $T^*(Q \times \mathfrak{g}  )$. We obtain
\begin{equation}\label{can_Ham_equ}
\dot \alpha = X_H( \alpha )= \xi _{T^* Q}( \alpha ) - \operatorname{Ver}_{ \alpha } \frac{\delta  h}{\delta  q} ,\quad 
\dot \xi = \frac{\delta  H}{\delta  \pi  } =\frac{\delta  h}{\delta  \pi  } ,\quad \dot \pi  =- \frac{\delta  H}{\delta \xi } =  \frac{\delta  \ell}{\delta \xi } - \mathbf{J} ( \alpha ).
\end{equation}
Clearly, these equations coincide with the stationarity conditions \eqref{stationarity_conditions_Clebsch}. In particular, the last equation reads
\[
\mathbf{J} ( \alpha ) =\frac{\delta  \ell}{\delta \xi }-\partial _t \frac{\delta  \ell}{\delta \dot \xi }.
\]

\subsubsection*{Example 1: Geodesic $2$-spline equation on Lie groups}
Recall from \S\ref{2_splines} that the reduced Lagrangian for $2$-splines on a Lie group $G$ with right $G$-invariant Riemannian metric reads
\begin{equation}\label{Lag-Riemann}
\ell(\xi,\dot{\xi})=\frac{1}{2} \|\eta\|^2=\frac{1}{2} \left \|\dot{\xi} + {\rm ad}^\dagger_\xi\xi\right\|^2.
\end{equation}
Here we denote 
\begin{equation}\label{ad-dagger}
\eta = \dot{\xi} + {\rm ad}^\dagger_\xi\xi
\quad\hbox{with}\quad
\operatorname{ ad}^\dagger_ \xi \nu =\left( \operatorname{ ad} ^\ast _\xi (\nu ^ \flat) \right) ^\sharp
\quad\hbox{for}\quad
\xi,\nu\in\mathfrak{g}.
\end{equation}
Then the quantity computed in equation \eqref{2nd-EPeqns2}
\begin{equation*}
\mu := \frac{\delta \ell}{\delta \xi} - \partial _t \frac{\delta \ell}{\delta \dot{\xi} }
=\left( 
\partial_t \eta
+ {\rm ad}_\eta \xi
+ {\rm ad}^\dagger_\eta \xi
\right) ^\flat
,\quad\hbox{with}\quad
\eta = \dot{\xi} + {\rm ad}^\dagger_\xi\xi
,\label{spline-mom}
\end{equation*}
satisfies the $2^{nd}$-order Euler-Poincar\'e equation,  $(\partial_t + {\rm ad}^*_\xi)\mu =0$, which is also the geometric 2-spline equation of \cite{CrSL1995}. We now consider the canonical formulation of $2$-splines.

\paragraph{Hamiltonian formulation of the geodesic $2$-spline equation on $T^*(Q \times \mathfrak{g})$.}
As we have seen, the Clebsch-Pontryagin approach of Section \ref{Clebsch-sec} allows the geodesic 2-spline equation to be recast as a set of canonical Hamilton equations for a Hamiltonian $H: T^*(Q \times \mathfrak{g}  ) \rightarrow  \mathbb{R}$. Note that in the case of $2$-splines, the variable $ \pi $ is
\[
\pi = \frac{\delta \ell }{\delta \dot \xi  } =\dot{\xi}^\flat + {\rm ad}^*_\xi\xi^ \flat= \eta ^\flat
\]
which proves that $\dot  \xi \mapsto \frac{\delta \ell }{\delta \dot \xi  }$ is a diffeomorphism. One thus obtains the Hamiltonian
\begin{align}\label{can_Ham_splines}
H( \alpha , \xi , \pi )&= \langle \pi , \dot \xi \rangle -\ell( \xi , \dot \xi )+ \left\langle \mathbf{J} ( \alpha ), \xi \right\rangle ,\qquad \qquad 
\pi = \frac{\delta \ell }{\delta \dot \xi  }\nonumber\\
&=\frac{1}{2} \| \pi \|^2- \left\langle  \pi, {\rm ad}^\dagger_\xi\xi \right\rangle +\left\langle \mathbf{J} ( \alpha ), \xi \right\rangle,
\end{align}
where $\|\cdot\|$ denotes the norm induced by $ \gamma $ on $ \mathfrak{g}  ^\ast $. The canonical Hamiltonian  formulation \eqref{can_Ham_equ} now yields the dynamical system
\begin{equation}\label{can_Ham_equ_2_splines}
\dot \alpha = \xi _{T^* Q}( \alpha )  ,\quad 
\dot \xi = \pi^\sharp - {\rm ad}^\dagger_\xi\xi ,\quad \dot \pi  = -{\rm ad}^*_{\pi^\sharp}\xi^\flat - \left({\rm ad}_{\pi^\sharp}\xi\right)^\flat-  \mathbf{J} ( \alpha  )
\end{equation}
As we have proved above, the {Euler-Poincar\'e} equation $(\partial_t + {\rm ad}^*_\xi)\mu =0$ is then established by noticing that $ \mu  =   \mathbf{J} (\alpha )$ is the cotangent-lift momentum map for the action of the Lie group $G$ on the manifold $Q$ 
and that the $\dot{\xi}$-equation implies $\pi^\sharp = \dot{\xi} + {\rm ad}^\dagger_\xi\xi  = \eta$.
Note that the solution for the momentum map $\mu =   \mathbf{J} ( \alpha )$ may be obtained entirely from the canonical Hamilton equations, without explicitly solving the Euler-Poincar\'e equation. For a bi-invariant metric, one has ${\rm ad}^\dagger_\xi\xi=0$ in the $\dot{\xi} $-equation and the last two terms cancel each other  in the $\dot{\pi} $-equation. Consequently, these two canonical equations simplify to $\dot{\xi} =   \pi^\sharp$ and $\dot{\pi} =  -  \mathbf{J} ( \alpha )$. From them, we find 
\begin{equation}\label{2_spline_bi_inv}
\ddot{\xi} = -  \mathbf{J} (\alpha  )^{\sharp}
\quad\hbox{and}\quad
\dddot{\xi} = -  {\rm ad}_\xi^\dagger\ddot{\xi}
,\end{equation}
in agreement with equation \eqref{CrSLe-commutator} and reference \cite{CrSL1995}.

\subsubsection*{Example 2: Geodesic $2$-spline equations on $SO(3)$}

We consider the particular case of the Lie group $G=SO(3)$ endowed with the bi-invariant metric induced by the standard  Ad-invariant inner product
\[
\gamma ( \Omega , \Gamma )= - \frac{1}{2} \operatorname{Tr}(\Omega \Gamma ).
\]
We identify the dual $\mathfrak{so}(3)^*$ with $\mathfrak{so}(3)$ using $ \gamma $ so that $ \Omega^\flat = \Omega $. Using the hat map  $\,\widehat{\,}: \mathfrak{so}(3)\to\mathbb{R}^3$ (see \eqref{hat_map}), the Euler-Poincar\'e equation in \eqref{2_spline_bi_inv} reads
\begin{equation}
\dddot{\mathbf{\Omega}} - \mathbf{\Omega}\times \ddot{\mathbf{\Omega}} = 0 
,
\label{EP-eqn-ddu}
\end{equation}
which was first found in \cite{NoHePa1989}. The difference in sign from that paper arises here from the choice of reduction by right-invariance instead of left-invariance.

\paragraph{Canonical Hamilton equations on $T^*\mathbb{R}^3\times T^*\mathbb{R}^3$ for the NHP equation.} The Hamiltonian formulation of the NHP equation \eqref{EP-eqn-ddu} for geometric splines on $SO(3)$ with a bi-invariant metric may be obtained  in canonical variables $(\mathbf{\Omega},\boldsymbol{\pi},\mathbf{q}, \mathbf{p})\in T^*\mathbb{R}^3\times T^*\mathbb{R}^3$ from the Hamiltonian (see \eqref{can_Ham_splines}),
\begin{equation}
H(\mathbf{\Omega},\boldsymbol{\pi},\mathbf{q}, \mathbf{p}) = \frac{1}{2} \|\boldsymbol{\pi}\|^2 + \mathbf{\Omega}\cdot {\mathbf{q}}\times{\mathbf{p}}
.
\label{NHP-Ham}
\end{equation}
This corresponds to the choice $Q= \mathbb{R}  ^3 $ on which $SO(3)$ acts by matrix multiplication.

This Hamiltonian produces canonical equations of the form,
\begin{eqnarray*}
\dot{\mathbf{q}} = \frac{\delta H}{\delta \mathbf{p}} =   \mathbf{\Omega}\times \mathbf{q} 
,
&&   
\dot{\mathbf{p}} = - \frac{\delta H}{\delta \mathbf{q}} =   \mathbf{\Omega}\times \mathbf{p} 
,\\
\dot{\mathbf{\Omega}} =  \frac{\delta H}{\delta \boldsymbol{\pi}} =  \boldsymbol{\pi}
,
&& 
\dot{\boldsymbol{\pi}} = - \frac{\delta H}{\delta \mathbf{\Omega}}
= -  {\mathbf{q}}\times{\mathbf{p}}
.
\end{eqnarray*}
The $(\mathbf{q},\mathbf{p})$-equations here imply that $\boldsymbol{\mu}  =   \mathbf{J} (\mathbf{q},\mathbf{p} ) =  {\mathbf{q}}\times{\mathbf{p}}$ 
obeys the {Euler-Poincar\'e} equation  for right invariance,
\[
\left(\partial _t  + {\rm ad}^*_\mathbf{\Omega}\right)\boldsymbol{\mu} =0=\dot{\boldsymbol{\mu} }-\mathbf{\Omega}\times \boldsymbol{\mu}
,\]
which  results in the NHP equation (\ref{EP-eqn-ddu}) when we substitute $\boldsymbol{\mu} =-\ddot{\mathbf{\Omega}}$. 

The canonical Hamiltonian formulation of the NHP equation provides some insight into the interpretation of its constants of motion. For example, the Hamiltonian (\ref{NHP-Ham}) Poisson commutes with $|\mathbf{q}|^2,|\mathbf{p}|^2,(\mathbf{q}\cdot \mathbf{p})$, and $|\mathbf{q}\times \mathbf{p}|^2$, although only the last of these Poisson commutes with all the others. The Hamiltonian (\ref{NHP-Ham}) also Poisson commutes with the vector $\mathbf{K}=\mathbf{\Omega} \times \boldsymbol{\pi}- \mathbf{q} \times  \mathbf{p}\in\mathbb{R}^3$. Although the components of $\mathbf{K}$ satisfy Poisson bracket relations $\{K_1,K_2\}=K_3$ and cyclic permutations with each other, their sum of squares $K^2=K_1^2+K_2^2+K_3^2$ again Poisson commutes with all the others. The presence of the two constants of motion $|\mathbf{q}\times \mathbf{p}|^2$ and $K^2$ in Poisson involution allows symplectic reduction from six degrees of freedom to four, but the reduced system is still far from being 
integrable. The Hamiltonian conservation laws may be expressed in terms of $(\mathbf{\Omega},\dot{\mathbf{\Omega}},\ddot{\mathbf{\Omega}})\in T^{(2)}\mathbb{R}^3$ as 
\[
|\mathbf{q}\times \mathbf{p}|^2=|\ddot{\mathbf{\Omega}}|^2
,\quad
\mathbf{K}=\mathbf{\Omega} \times \dot{\mathbf{\Omega}} + \ddot{\mathbf{\Omega}}
,\quad
K^2=|\mathbf{\Omega} \times \dot{\mathbf{\Omega}}|^2 + 2 (\mathbf{\Omega} \times \dot{\mathbf{\Omega}})\cdot \ddot{\mathbf{\Omega}} + |\ddot{\mathbf{\Omega}}|^2
.\]
All of these conservation laws were known in the literature, but had previously not been given a Hamiltonian interpretation. 
The Hamiltonian interpretation of the NHP equation (\ref{EP-eqn-ddu}) in this setting is that the rotations act on the cross product $\mathbf{m}=\mathbf{q}\times \mathbf{p}$ diagonally in $\mathbf{q}$ and $\mathbf{p}$, so that $\dot{\mathbf{m}}=\mathbf{\Omega}\times \mathbf{m}$ for $\dot{\mathbf{q}} = \mathbf{\Omega}\times \mathbf{q}$ and $\dot{\mathbf{p}} = \mathbf{\Omega}\times \mathbf{p}$. This is also the essence of the symmetric representation of rigid body motion discussed, e.g., in \cite{BlCr1996}. 

\subsubsection*{Higher order case}

The canonical Clebsch formulation presented above can be adapted to higher order cost functions $\ell=\ell( \xi , \dot \xi , ..., \xi^{(k-1)} ,q)$ as follows. If $\xi^{(k-1)}\mapsto \frac{\delta  \ell}{\delta  \xi^{(k-1)}}$ is a diffeomorphism,  we define the function
\begin{align}\label{kth_order_h}
&h\left(  \xi , \dot \xi ,...,\xi^{(k-2)},  \pi _2,...,\pi _k,q \right) \nonumber\\
&\qquad := \langle \pi _2 , \dot \xi \rangle + \langle \pi _3 , \ddot \xi \rangle +...+ \left\langle \pi _k , \xi^{(k-1)} \right\rangle- \ell( \xi , \dot \xi ,...,\xi^{(k-1)} ,q),\qquad \pi _k = \frac{\delta \ell }{\delta  \xi^{(k-1)}  }
\end{align} 
and we consider the Hamiltonian $H: T^*(Q \times  (k-1) \mathfrak{g} )\rightarrow \mathbb{R}  $ given by
\[
H\left( \alpha _q , \xi ,..., \xi ^{(k-2)}, \pi _2 , ...,\pi _k \right) :=h\left(  \xi , \dot \xi ,...,\xi^{(k-2)},  \pi _2,...,\pi _k,q \right)+ \left\langle \mathbf{J} ( \alpha _q ), \xi \right\rangle.
\]
A straightforward computation shows that the canonical Hamilton equations on $T^*(Q \times (k-1) \mathfrak{g}  )$ for $H$ produce the stationarity condition \eqref{thm_can_form} of the $k^{th}$-order Clebsch-Pontryagin optimal control with cost function $\ell$ and therefore imply the $k^{th}$-order Euler-Poincar\'e equations \eqref{EP_with_q_RL}. We thus obtain the generalization of Theorem \ref{thm_can_form} for $k^{th}$-order cost functions.

\medskip

\subsection{Ostrogradsky-Lie-Poisson reduction}

The procedure of Lie-Poisson reduction of the Hamilton-Ostrogradsky theory parallels that for higher-order Euler-Poincar\'e reduction and produces a different Hamiltonian formulation of the higher-order dynamics that applies to 
$k\ge2$. At first, we will discuss the Hamilton-Ostrogradsky approach for the higher-order Hamiltonian formulation based purely on Lie group reduction, i.e., without introducing the action of the Lie group $G$ on the manifold $Q$. Then we will remark on how $q$-dependence may be easily incorporated.

\paragraph{Second order.}
Consider a $G$-invariant second order Lagrangian
$L:T^{(2)}G \rightarrow \mathbb{R}  $, $L=L(g, \dot g, \ddot g)$. The Ostrogradsky momenta are defined by the fiber derivatives,
\[
p _1 = \frac{\partial  L}{\partial  \dot g} - \partial _t  \frac{\partial  L}{\partial  \ddot g} , \qquad p _2 = \frac{\partial  L}{\partial  \ddot g}
\]
and produce the Legendre transform $(g, \dot g, \ddot g, \dddot{g})\in T^{(3)}G \mapsto (g, \dot g, p _1 , p _2 )\in T^*(TG)$. We refer to \cite{deLRo1985} for the intrinsic definition of the Legendre transform for higher order Lagrangians. See also \cite{BlCr1996a} for an application on $SO(3)$. 

When the Legendre transform is a diffeomorphism, the corresponding Hamiltonian $H:T^*(TG) \rightarrow \mathbb{R}  $ is defined by
\[
H(g, \dot g, p _1 , p_2 ):=\left\langle p _1 , \dot g \right\rangle + \left\langle p _2 , \ddot g \right\rangle -L(g, \dot g, \ddot g)
\]
and the canonical Hamilton equations for $H$ are equivalent to the $2^{nd}$-order Euler-Lagrange equations for $L$.

Applying reduction by symmetry to $H$ induces a Hamiltonian $h(\pi _1,\xi , \pi _2 )$ on $T^*(TG)/G\simeq \mathfrak{g} ^\ast \times  T^* \mathfrak{g} $, which is related to the symmetry-reduced Lagrangian $\ell(\xi , \dot \xi )$ by the corresponding Legendre transformation,
\begin{equation}
h( \pi _1 ,\xi ,  \pi _2 )=\Big\langle \pi  _1 , \xi  \Big\rangle 
+ \left\langle \pi  _2 , \dot \xi  \right\rangle 
-\ell(\xi , \dot \xi )
,\quad \frac{\delta  \ell}{\delta \dot \xi } = \pi _2.
\label{h-def-LegXform}
\end{equation}

By Reduction of the Hamilton-Ostrogradsky equations for $H$ on $T^*(TG)$ we obtain the \textit{Ostrogradsky-Lie-Poisson} equations for $h$
\begin{framed}
\begin{equation}\label{Ostro_LP}
\left\{\begin{array}{l}
\displaystyle \partial _t  \pi _1 \pm \operatorname{ad}^*_{\frac{\delta  h}{\delta \pi _1 } }\pi _1 =0,\\
\displaystyle\partial _t \xi = \frac{\delta  h}{\delta \pi _2  },\qquad 
\partial _t \pi _2  = -\frac{\delta  h}{\delta \xi },
\end{array}\right.
\end{equation}
\end{framed}\noindent
together with the non-canonical Poisson bracket given by
\begin{align}\label{tot_Poiss_2}
\{f,g\}( \pi _1, \xi ,  \pi_2)
&=\pm \left\langle \pi _1 , \left[ \frac{\delta f }{\delta \pi_1},\frac{\delta g }{\delta \pi_1} \right] \right\rangle +\left\langle \frac{\delta f }{\delta \xi  }, \frac{\delta  g}{\delta \pi _2  }  \right\rangle-\left\langle \frac{\delta g }{\delta \xi  }, \frac{\delta  f}{\delta \pi _2  }  \right\rangle\nonumber\\
& = \{f,g\}_{\pm}(\pi _1)+ \{f,g\}_{can}(\xi,\pi _2)
\end{align}
for functions $f,g$ depending on the variables $( \pi _1,\xi  , \pi _2 )$.
Note that this reduction process holds without assuming a preexisting Lagrangian formulation. Equations \eqref{Ostro_LP} together with their Hamiltonian structure can be obtained by Poisson reduction for cotangent bundles: $T^*Q \rightarrow T^*Q/G$ (the so called Hamilton-Poincar\'e reduction \cite{CeMaPeRa2003}) applied here to the special case $Q=TG$.

We now check directly that equations \eqref{Ostro_LP} are equivalent to the $2^{nd}$-order Euler-Poincar\'e equations
if the Hamiltonian \eqref{h-def-LegXform} is associated to 
$\ell$ by an invertible Legendre transform.
The derivatives of the symmetry-reduced Hamiltonian $h$ with respect to the momenta $\pi_1$ and $\pi_2$  imply formulas for the velocity and acceleration,
\begin{equation}
\frac{\delta  h}{\delta  \pi _1 } = \xi  ,\quad
\frac{\delta  h}{\delta  \pi _2 }  =\dot \xi 
,
\label{2ndOrderLP-eqns}
\end{equation}
so that the acceleration $\dot \xi$ may be expressed as a function of the velocity $\xi$ and the momenta $(\pi_1,\pi_2)$. 
The pair $(\xi,\pi _2)\in T^*\mathfrak{g}\simeq \mathfrak{g}\times \mathfrak{g}^*$ obeys canonical Hamilton equations, so the derivatives of $h$ in (\ref{h-def-LegXform}) with respect to velocity and acceleration  imply the momentum relations,
\begin{equation}
\dot \pi_2 = - \frac{\delta  h}{\delta \xi } = \frac{\delta \ell }{\delta \xi  } - \pi _1 
\quad\text{and}\quad 
\frac{\delta  h}{\delta \dot{\xi} }= 0 =
\pi _2 - \frac{\delta \ell }{\delta\dot \xi   } 
.
\label{2ndOrderVar-eqns}
\end{equation}
Solving these momentum relations for $\pi_1$ and $\pi_2$ in terms of derivatives of the reduced Lagrangian yields
\begin{equation}
\pi_1 = \frac{\delta  \ell}{\delta \xi  } - \partial _t  \frac{\delta \ell }{\delta \dot{ \xi } }
\quad\text{and}\quad 
\pi_2 = \frac{\delta  \ell}{\delta \dot \xi  }
.
\label{pi1-pi2-defs}
\end{equation}
The Lie-Poisson equation for $\pi_1$,
\begin{equation}
\partial _t  \pi _1 \pm \operatorname{ad}^*_{\frac{\delta  h}{\delta \pi _1 } }\pi _1 =0
,
\label{pi1-LP}
\end{equation}
then implies the $2^{nd}$-order {Euler-Poincar\'e} equation,
\[
\left( {\partial _t}\pm {\rm ad}_\xi^* \right) 
 \left( 
\frac{\delta  \ell}{\delta \xi  } - \partial _t \frac{\delta \ell }{\delta \dot{ \xi } }
\right)=0
.
\]

\subsubsection*{Example: Ostrogradsky-Lie-Poisson approach for geometric 2-splines}

The Ostrogradsky reduced Hamiltonian \eqref{h-def-LegXform} for geometric 2-splines is
\begin{equation}
h(\pi_1,\xi,\pi_2) = \frac12\left\|\pi_2\right\|^2- \Big\langle \pi_2, {\rm ad}_\xi^\dagger \xi \Big\rangle
+ \Big\langle \pi_1, \xi \Big\rangle
.
\label{Ost-ham-splines}
\end{equation}
From this reduced Hamiltonian, the Poisson bracket \eqref{tot_Poiss_2} recovers the geometric 2-spline equations \eqref{2nd-EPeqns2}. For a bi-invariant metric $ {\rm ad}_\xi^\dagger \xi = 0$ and these equations reduce to \eqref{CrSLe-commutator}. In addition for $SO(3)$ these equations produce the NHP equation \eqref{EP-eqn-ddu}. 

\paragraph{Third order.}
Before going to the general case, it is instructive to quickly presenting the case of a  third order $G$-invariant Lagrangian  $L:T^{(3)}G \rightarrow \mathbb{R}  $, $L=L(g, \dot g, \ddot g, \dddot g)$ inducing the symmetry-reduced Lagrangian $\ell: T^{(3)}G/G\simeq 3\mathfrak{g}  \rightarrow \mathbb{R}$, $\ell=\ell( \xi , \dot \xi , \ddot \xi )$. The Ostrogradsky momenta
\[
p _1 = \frac{\partial  L}{\partial  \dot g} - \partial _t  \frac{\partial  L}{\partial  \ddot g} 
+ \partial ^2 _t  \frac{\partial  L}{\partial  \dddot g}
, \qquad
p _2 = \frac{\partial  L}{\partial  \ddot g} - \partial _t  \frac{\partial  L}{\partial  \dddot g}
,\qquad 
p _3  =  \frac{\partial  L}{\partial  \dddot g},
\]
produce the Legendre transform $(g, \dot g,..., g^{(5)} )\in T^{(5)}G \mapsto (g, \dot g, \ddot g, p _1 , p _2 , p _3 )\in T^*(T^{ (2) }G)$.
The associated $G$-invariant Hamiltonian is obtained from the Legendre transformation,
\[
H:T^*(T^{(2)}G) \rightarrow \mathbb{R}  ,\quad H(g, \dot g, \ddot g, p _1 , p _2 , p _3 )
:=\left\langle p _1 , \dot g \right\rangle + \left\langle p _2 , \ddot g \right\rangle
+ \left\langle p _3 , \dddot g \right\rangle  -L(g, \dot g, \ddot g, \ddot g),
\]
so that $G$-invariance of the Hamiltonian $H$ yields the symmetry-reduced Hamiltonian $h(\pi_1 ,\xi ,\dot{\xi}, 
\pi _2, \pi _3)$, $h:T^*(T^{(2)}G)/G\simeq  \mathfrak{g}  ^\ast \times T^*(2\mathfrak{g}) \rightarrow \mathbb{R}$. 
The reduced Hamiltonian $h$ is related to the reduced Lagrangian $\ell$ by the extended Legendre transformation
\begin{equation}
h (  \pi _1 ,\xi ,\dot \xi , \pi _2, \pi _3  )=\Big\langle \pi  _1 , \xi  \Big\rangle 
+ \left\langle \pi  _2 , \dot \xi  \right\rangle 
+\left\langle \pi _3 , \ddot \xi  \right\rangle  -\ell(\xi , \dot \xi , \ddot \xi )
,\qquad \frac{\delta \ell }{\delta\ddot \xi   } = \pi _3.
\label{3rdOrderRedHam}
\end{equation}
The $3^{rd}$-order Ostrogradsky-Lie-Poisson system reads
\begin{framed}
\begin{equation}\label{Ostro_LP_3rd}
\left\{\begin{array}{l}
\displaystyle \partial _t  \pi _1 \pm \operatorname{ad}^*_{\frac{\delta  h}{\delta \pi _1 } }\pi _1 =0
,\\
\displaystyle\partial _t \xi = \frac{\delta  h}{\delta \pi _2  },\qquad 
\partial _t \pi _2  = -\frac{\delta  h}{\delta \xi },\qquad \partial _t \dot \xi = \frac{\delta  h}{\delta \pi _3},\qquad 
\partial _t \pi _3 = -\frac{\delta  h}{\delta \dot \xi },
\end{array}\right.
\end{equation}
\end{framed}\noindent
with associated Poisson bracket
\begin{equation}\label{PB}
\{f,g\}( \pi _1 , \xi ,\dot \xi , \pi _2 , \pi _3  )
=\{f,h\}_{\pm}(\pi _1 )+\{f,h\}_{can}(\xi ,\pi _2) +\{f,h\}_{can}(\dot \xi , \pi _3 ).
\end{equation}

If the Hamiltonian \eqref{3rdOrderRedHam} is associated to a Lagrangian $\ell$ by Legendre transformation, we have
\begin{equation}
\dot \xi = \frac{\delta  h}{\delta  \pi _2 }
 ,\quad 
 \dot \pi_2 = -\frac{\delta  h}{\delta \xi }=- \pi _1 + \frac{\delta \ell }{\delta \xi  } 
 ,\quad 
 \ddot \xi = \frac{\delta  h}{\delta  \pi _3} 
 ,\quad 
 \dot \pi_3 = -\frac{\delta  h}{\delta \dot \xi }= - \pi _2 + \frac{\delta  \ell}{\delta \dot \xi  } 
 ,\quad 
 \dot \pi _1 + \operatorname{ad}^*_{\frac{\delta  h}{\delta \pi _1 } }\pi _1 =0
 .
 \label{redHam-eqns}
\end{equation}
Consequently, we have 
\[
\pi _1 =\frac{\delta \ell }{\delta \xi  }- \dot \pi_2 =\frac{\delta \ell }{\delta \xi  }-\partial _t \left(\frac{\delta  \ell}{\delta \dot \xi  }- \dot \pi _3  \right)= \frac{\delta \ell }{\delta \xi  }-\partial _t \frac{\delta  \ell}{\delta \dot \xi  } + \partial _t ^2 \frac{\delta  \ell}{\delta \ddot \xi  }
,
\]
and the last equation in \eqref{redHam-eqns} for the momentum map $\pi_1$ implies by the 3rd-order {Euler-Poincar\'e} equation
\[
\left( {\partial _t}\pm {\rm ad}_\xi^* \right) 
 \left( 
\frac{\delta  \ell}{\delta \xi  } - \partial _t \frac{\delta \ell }{\delta \dot{ \xi } }+ \partial _t ^2 \frac{\delta  \ell}{\delta \ddot \xi  }
\right)=0
.
\]

\paragraph{Higher-order and $q$-dependence.} The Ostrogradsky-Lie-Poisson approach generalizes to $k^{th}$-order as follows. For a $G$-invariant Lagrangian $L:T^{(k)}G \rightarrow \mathbb{R}  $, $L=L(g, \dot g, ..., g^{(k)})$, the Ostrogradsky momenta define the Legendre transform as a map $T^{(2k-1)}G \rightarrow T^*(T^{(k-1)}G)$ (see \cite{deLRo1985}) and the associated Hamiltonian $H: T^*(T^{(k-1)}G) \rightarrow \mathbb{R}$, $H=H(g, \dot g, ..., g^{(k-1)}, p _1 , ..., p_{k})$ is given by
\begin{equation}
H(g, \dot g, ..., g^{(k-1)}, p _1 ,\ldots , p_k)
:=\sum _{j=1}^{k}\left\langle g^{(j)}, p_j \right\rangle 
-L(g, \dot g, ..., g^{(k)}).
\end{equation}
Extending the symmetry-reduced Ostrogradsky procedure outlined above to $k^{th}$-order then leads to the Ostrogradsky-Lie-Poisson equations
\begin{framed}
\begin{equation}\label{Ostro_LP_kth}
\left\{\begin{array}{l}
\displaystyle \vspace{0.2cm}\partial _t  \pi _1 \pm \operatorname{ad}^*_{\frac{\delta  h}{\delta \pi _1 } }\pi _1 =0
,\\
\displaystyle\partial _t \xi^{(j-2)} = \frac{\delta  h}{\delta \pi _{j}  },\qquad 
\partial _t \pi _{j}  = -\frac{\delta  h}{\delta \xi^{(j-2)} },\qquad j=2,...,k,

\end{array}\right.
\end{equation}
\end{framed}\noindent
whose Hamiltonian structure is
\begin{equation}\label{PB-k}
\{f,g\}( \pi _1 ,\xi,\dots , \xi^{(k-2)} ,\pi_2,...,\pi _{k})
= \{f,g\}_{\pm}(\pi _1 )+\sum_{j=2}^{k}\{f,g\}_{can}(\xi^{(j-2)} , \pi _{j})
.
\end{equation}
This Poisson bracket produces the geometric $k$-spline equations from the corresponding reduced Hamiltonian; we
do not carry out the details here.

The Ostrogradsky procedure for reduction by symmetry outlined above generalizes to allow $q$-dependence. Indeed, for a $G$-invariant Lagrangian $L:\left( T^{(k)}G \right) \times Q \rightarrow \mathbb{R}  $, the previous steps may all be repeated with only slight changes, resulting in the reduced Poisson bracket \eqref{PB-k}, modified by adding the terms
\begin{equation}
\left\langle \mathbf{J} \left( \frac{\delta f }{\delta q } \right) , \frac{\delta  h}{\delta  \pi _1 }\right\rangle
-
\left\langle \mathbf{J} \left( \frac{\delta h }{\delta q } \right) , \frac{\delta  f}{\delta  \pi _1 }\right\rangle .
\label{momap-brktmod}
\end{equation}
Thus, allowing $q$-dependence leaves the canonical equations invariant, but alters the $\pi_1$-equation so that it becomes
\begin{equation}
 \partial _t \pi _1 \pm   \operatorname{ad}^*_{\frac{\delta  h}{\delta \pi _1 } }\pi _1= - \mathbf{J} \left( \frac{\delta h }{\delta q }\right)
.
\label{pi1-LP+q}
\end{equation}

\paragraph{From Clebsch to Ostrogradsky-Lie-Poisson.} Recall that, in the case $k=1$, one passes from the canonical Clebsch formulation to the Lie-Poisson formulation (with Lie-Poisson bracket \eqref{red_Poisson}) by using the momentum map, via the transformation
\[
\alpha_q  \in T^*Q \mapsto ( \mu , q)\in \mathfrak{g}  ^\ast \times Q,\quad  \mu = \mathbf{J} ( \alpha _q ).
\]
Consider now the case $k=2$ with Lagrangian $\ell=\ell(\xi , \dot \xi , q)$. On the Clebsch side, the Hamiltonian is given by
\[
H( \alpha _q , \xi , \pi ):= \langle \pi , \dot \xi \rangle -\ell( \xi , \dot \xi ,q)+ \left\langle \mathbf{J} ( \alpha_q  ), \xi \right\rangle ,\quad \pi = \frac{\delta \ell }{\delta \dot \xi  }\,,
\]
whereas the Ostrogradsky-Lie-Poisson Hamiltonian with $q$-dependence is defined by
\[
h(\pi _1 , \xi ,\pi _2,q ):=\left\langle \pi _1 , \xi \right\rangle +\left\langle \pi _2 , \dot \xi \right\rangle - \ell( \xi , \dot \xi ,q),\quad \pi_2  = \frac{\delta \ell }{\delta \dot \xi  }\,.
\]
These definitions suggest that one can pass from the canonical Clebsch formulation on $T^*(Q\times \mathfrak{g} )$ to the Lie-Poisson-Ostrogradsky formulation on $\mathfrak{g}  ^*\times T^* \mathfrak{g}  \times Q$ by the Poisson map
\[
(\alpha _q ,\xi , \pi _2 ) \mapsto (\pi _1 , \xi  ,\pi _2, q ),\quad \pi _1 := \mathbf{J} ( \alpha _q ).
\]
This is indeed the case, as one can check easily. 
Generalization to $k>2$ is now straightforward and we have the Poisson map
\[
T^*(Q \times (k-1) \mathfrak{g}  ) \rightarrow \mathfrak{g}  ^* \times T^*(k-1) \mathfrak{g}  \times Q,
\]
\[
( \alpha _q ,\xi , \dot \xi ,..., \xi ^{(k-2)},\pi _2 ,..., \pi _k )\mapsto (\pi _1 , \xi , \dot \xi ,..., \xi ^{(k-2)},\pi _2 ,..., \pi _k,q),\quad \pi _1 := \mathbf{J} ( \alpha _q ),
\]
that relates the canonical Hamilton equations on $T^*(Q \times ( k-1) \mathfrak{g}  )$ and the Ostrogradsky-Lie-Poisson equations \eqref{Ostro_LP_kth}.

\section{Outlook and open problems} \label{outlook-sec}

\subsection{Brief summary and other potential directions} 

This paper has begun the application of symmetry-reduction tools to higher-order variational problems on Lie groups, culminating in an application of $2^{nd}$-order geometric splines to template matching on the sphere which was shown to be governed by a higher-order Euler-Poincar\'e equation on the dual Lie algebra of the Lie group $SO(3)$. The generality of this result was emphasized in Remark \ref{HOEP-remark}, on seeing that the higher-order Euler-Poincar\'e equation (\ref{TM_HO_EP}) had emerged once again as the optimality condition for template matching. 

Various open problems not treated here seem to crowd together to present themselves. A few of these are:

\begin{itemize}

\item
We have applied variational constraints to $k$-splines in Section \ref{subsec_constraints}. However,  accommodating nonholonomic constraints would require additional developments of the theory.

\item
In Section \ref{hoTemplateMatching-sec} we presented higher-order methods that increased the smoothness in interpolating through a sequence of data points. In future work these methods will be compared to the shape spline model introduced in \cite{TrVi2010}. Some initial forays into the analysis of these problems were also presented in Section \ref{some-analytical-remarks-sec}, but much remains to be done for these problems that have been treated here only formally. 

\item
Extension of the basic theory presented here to allow for actions of Lie groups on Riemannian manifolds is also expected to have several interesting applications, particularly in image registration. 
For example, one could address higher-order Lie group invariant variational principles that include both curves on Lie groups and the actions of Lie groups on smooth manifolds, particularly on Riemannian manifolds.  These Lie group actions on manifolds apply directly to the optimal control problems associated with large-deformation image registration in the \textit{Large Deformation by Diffeomorphisms Metric Mapping} (LDDMM) framework via Pontryagin's maximum principle. Actions of Lie groups on Riemannian manifolds will be investigated in a subsequent treatment.

\end{itemize}

Doubtlessly, other opportunities for applying and extending this symmetry reduction approach for $k$-splines will present themselves in further applications.

\subsection{An open problem:  the slalom, or brachistochrone for splines}

Let us formulate yet another example in slightly more detail. This is the brachistochrone version of the optimization problem treated here, for possible application, say, in a slalom race. 
Unlike the optimization problem, which seeks the path of least \emph{cost}, a race would seek the path of least \emph{time}. 
For example, the familiar \emph{slalom} race involves dodging around a series of obstacles laid out on the course. The objective of the slalom racer in down-hill skiing, for example, is to pass through a series of gates as quickly as possible. The strategy in slalom racing is to stay close to the shortest-time path (or geodesic) between the gates, while also moderating the force exerted in turning to keep it below some threshold, lest the snow give way and the skier slides off the course. Thus, the ideal slalom path sought by an expert racer would hug the geodesic between the gates and make the series of turns passing through the gates with no skidding at all.

The strategy for achieving the optimal slalom has many potential applications in modern technology. 
For example:
\begin{itemize}
\item
 A charged-particle beam in an accelerator may be guided in its path by a series of quadrupole magnets that steer the beam to its target. The steering must be done as gently as possible, so as to minimize the transverse acceleration (seen as curvature in the path of the beam) that causes Bremsstrahlung and the consequent loss of energy in the beam. 
\item
An underwater vehicle may be steered smoothly through narrow passageways in a sunken ship along a path that will take it quickly and efficiently to its objective, thereby avoiding collisions while minimizing fuel expenditure, time, etc.
\item
A car may be programmed to glide smoothly through a tight parallel-parking maneuver that ends in an elegant stop in the narrow space between two other cars along the curb. 
\item
A vehicle may follow a program to roll as rapidly as possible along the terrain through a series of gates with its cameras mounted so that they continuously point toward an object above it that must keep in sight.
\end{itemize}

The slalom strategy that applies in all these examples seeks a path that minimizes the time for the distance travelled over a prescribed course, while also moderating the acceleration or force exerted along the path as it passes around a series of obstacles or through a series of gates laid out on the course. Designing such maneuvers requires optimization for least time, while also using cost functions that depend on both the velocity and acceleration of the motion. Moderation of higher-order accelerations such as \emph{jerk} (rate of change of acceleration) may also be needed. As in the present paper,  in solving optimal slalom problems that minimize the time taken to finish the course, one might expect to take advantage of continuous symmetries by investigating Lie group invariant variational problems for cost functions that are defined on $k^{th}$-order tangent spaces of Lie groups acting on smooth Riemannian manifolds. Investigations of invariant variational principles 
for slalom problems using the present group reduction and induced metric methods would be a promising direction. However, this direction seems even more challenging than the geometric splines for optimizing costs in trajectory planning on a Lie group treated here and it will be deferred to a later paper.

\subsection*{Acknowledgements}
The work of FGB was partially supported by a Swiss NSF postdoctoral fellowship.
DDH is grateful for partial support by the Royal Society of London, Wolfson Scheme. 
TSR was partially supported by  Swiss NSF grant 200020-126630.  
We also thank M. Bruveris and A. Trouv\'e for encouraging comments and insightful remarks during the course of this work.

\bibliographystyle{alpha}
\bibliography{myrefs-FGBDM}

\end{document}